\documentclass[11pt]{article}

\usepackage{amsmath,amssymb,amsthm,amsfonts,mathtools}
\usepackage{bm}
\usepackage{bbm}
\usepackage{graphicx}
\usepackage{booktabs}
\usepackage{enumitem}
\usepackage{hyperref}
\usepackage{natbib}
\usepackage{color}
\usepackage{algorithm}
\usepackage{algpseudocode}
\newenvironment{practitionerSummary}{%
	\begin{abstract}%
	}{%
	\end{abstract}%
}
\usepackage{subcaption}
\newtheorem{theorem}{Theorem}[section]
\newtheorem{proposition}[theorem]{Proposition}

\theoremstyle{definition}

\theoremstyle{remark}

\newcommand{\R}{\mathbb{R}}
\newcommand{\E}{\mathbb{E}}
\newcommand{\Prob}{\mathbb{P}}
\newcommand{\1}{\mathbbm{1}}

\newcommand{\Var}{\mathrm{Var}}

\newcommand{\argmin}{\mathop{\mathrm{arg\,min}}}
\newcommand{\argmax}{\mathop{\mathrm{arg\,max}}}

\begin{document}

	\title{Entropy Regularization under No-Shorting: Structure Preservation and Convexity in Dynamic Mean–Variance Portfolio Selection}
	\author{
		XIANGYU CUI\\
		{\small School of Statistics and Data Science, Dishui Lake Advanced Finance Institute,}\\
		{\small Shanghai University of Finance and Economics, Shanghai, China. Email: cui.xiangyu@mail.shufe.edu.cn.}\\
		XUN LI\\
		{\small Department of Applied Mathematics,}\\
		{\small The Hong Kong Polytechnic University, Hong Kong, China. Email: li.xun@polyu.edu.hk.}\\
		ZIYAN LI\\
		{\small School of Statistics and Data Science,}\\
		{\small Shanghai University of Finance and Economics, Shanghai, China. Email: liziyan1204@qq.com}\\		
		YUN SHI\\
		{\small School of Statistics, }\\
		{\small East China Normal University, Shanghai, China. Email: yshi@fem.ecnu.edu.cn.}\\
		}
	\date{}
\maketitle

\begin{abstract}
This paper studies entropy regularization for the discrete-time dynamic mean--variance portfolio problem with no-shorting constraints. In the unconstrained setting, entropy regularization can be introduced without changing the basic quadratic structure of the dynamic program. Under no-shorting, this property is lost: a standard fixed-temperature formulation generally preserves convexity but no longer preserves the homogeneous Bellman structure and the associated piecewise-linear feedback form. We introduce a state-dependent temperature that is matched to the homogeneity of the original mean--variance problem. This restores the piecewise-quadratic value function and the finite-dimensional Bellman recursion. The resulting formulation, however, makes convexity depend on the continuation coefficients generated by the soft recursion. We show that negative continuation coefficients lead to loss of convexity and prove that a common sufficiently small temperature preserves the Bellman structure and convexity over the entire finite horizon. The original no-shorting solution is recovered as the temperature vanishes. We also develop a sample-based implementation for unknown return distributions. Numerical and empirical results illustrate both the structural advantage of the state-dependent temperature and the practical relevance of the convexity-preserving temperature restriction.
\end{abstract}

\begin{practitionerSummary}
This study develops a multi-period portfolio method for investors who face no-shorting constraints and wish to use entropy regularization without losing the tractable structure of the original mean--variance problem. A standard fixed temperature can preserve convexity but generally destroys the simple Bellman recursion under no-shorting. The proposed state-dependent temperature retains the piecewise-linear feedback form and provides a clear condition for when the backward recursion remains convex. This condition is practically relevant because an excessively large temperature can invalidate the recursion even when the optimization problem is otherwise well defined. The method can be implemented directly from historical return data and produces a deterministic portfolio rule through the mode of the resulting policy distribution. The framework is therefore useful for constrained dynamic asset allocation when both computational tractability and control of regularization strength are important.
\end{practitionerSummary}

\section{Introduction}
\label{sec:introduction}
Mean--variance portfolio selection originates from the classical work of
\citet{markowitz1952portfolio}. Extending the static formulation to a
multi-period setting is nontrivial because terminal variance is not
additively separable over time. A standard approach embeds the problem
into a family of auxiliary quadratic control problems, as in the
continuous-time stochastic linear--quadratic formulation of
\citet{zhou2000continuous} and the discrete-time embedding of
\citet{li2000optimal}. The latter restores dynamic programming by
introducing a multiplier for the terminal mean constraint.

Portfolio constraints substantially alter this structure. Under
no-shorting, risky positions are restricted to the nonnegative orthant,
so the unconstrained linear feedback rule and quadratic recursion no
longer apply directly. In continuous time, \citet{li2002dynamic}
characterize the mean--variance solution under no-shorting, while
\citet{czichowsky2013cone} study Markowitz problems under general cone
constraints. In discrete time, \citet{cui2014optimal} show that the value
function becomes piecewise quadratic and the optimal policy piecewise
linear; related cone-constrained formulations are developed in
\citet{cui2017mean}. The tractability of the discrete-time no-shorting
problem rests on a homogeneous Bellman structure: after the embedding
transformation, wealth enters through a quadratic scale factor, while the
remaining optimization reduces to a finite-dimensional backward
recursion.

Entropy regularization replaces a point-valued control by a
distributional decision rule and is widely used for exploration in
stochastic control and reinforcement learning.
\citet{wang2020stochasticcontrol} formulate continuous-time exploratory
control through entropy-regularized relaxed policies and show that the
linear--quadratic case admits Gaussian exploration.
\citet{tang2022exploratory} establish well-posedness and
zero-exploration convergence for the associated exploratory HJB
equations, while \citet{szpruch2024optimal} study the scheduling of the
entropy regularizer in continuous-time linear--quadratic reinforcement
learning. Mean--variance portfolio selection provides a natural financial setting
for this approach. \citet{wang2020mv} develop an entropy-regularized
continuous-time mean--variance formulation in which the exploratory
policy remains Gaussian. Related extensions include regime-switching
markets with partial information \citep{wu2024regime}, while
\citet{shi2026noshorting} develop a discrete-time entropy-regularized
mean--variance formulation. Entropy regularization has also been applied
to equilibrium mean--variance policies \citep{dai2023learning} and
models with jumps \citep{bender2023entropy}. These studies exploit the
tractability of quadratic portfolio problems, but randomized policies
become more delicate when portfolio constraints bind.

A recent continuous-time study by \citet{li2026noshorting} illustrates
this difficulty. Enforcing no-shorting on the support of a randomized
Gaussian policy destroys the unconstrained Gaussian exploration
structure. Their approach therefore uses an auxiliary exploratory
problem in which the mean of the Gaussian policy satisfies the
no-shorting restriction, while individual realizations need not. In
contrast, we require every realization of the randomized portfolio to
satisfy the constraint by supporting the policy distribution on the
nonnegative orthant. The resulting question is whether entropy
regularization can be introduced without destroying the Bellman
structure of the original constrained problem.

A standard fixed-temperature entropy regularization of the embedded
auxiliary problem does not, in general, preserve this structure. In the
unconstrained quadratic problem, integration over the full action space
retains Gaussian closure and hence a quadratic value function. Under
no-shorting, however, the soft Bellman integral is taken over the
nonnegative orthant, making its normalization state dependent. The value
function therefore generally leaves the homogeneous piecewise-quadratic
family, so the backward recursion becomes function valued rather than
being summarized by a finite number of continuation coefficients.

We restore the finite-dimensional recursion by matching the entropy
scale to the degree-two homogeneity of the embedded auxiliary problem.
On the positive-state branch, the risky allocation is a wealth-gap scale
times a state-independent portfolio coefficient. We randomize this
coefficient over the nonnegative orthant and set the entropy temperature
proportional to the squared wealth gap. This preserves the
piecewise-quadratic value function and keeps the relative strength of
regularization constant across positive states. The resulting soft
recursion, however, introduces entropy into the continuation coefficient
that determines curvature. If this coefficient becomes negative,
convexity is lost and the failure propagates backward. We show that
there exists a common positive temperature bound under which all
continuation coefficients remain positive and the entire finite-horizon
recursion remains well posed and convex. We also provide a directly
checkable sufficient condition and establish zero-temperature recovery
of the continuation coefficients, mode policy, and embedding multiplier.
These results identify both the structure-preserving entropy scaling and
the recursive convexity condition required for its validity.

For unknown return distributions, the proposed recursion admits a direct
sample-based implementation. Numerical experiments and a rolling application
to Fama--French industry portfolios show that the positivity condition defines
a nontrivial but finite admissible temperature region.

The remainder of the paper is organized as follows. 
Section~\ref{sec:structural-challenge} reviews the deterministic no-shorting recursion and shows why a fixed entropy temperature loses the original Bellman closure. 
Section~\ref{sec:structure-preserving-soft} develops the structure-preserving entropy formulation and establishes the whole-horizon convexity result. 
Section~\ref{sec:sample-implementation} presents the sample-based implementation. 
Section~\ref{sec:numerical-analysis} reports the numerical analysis, and Section~\ref{sec:empirical-illustration} presents the rolling empirical application. 
The final section concludes.

\section{Fixed-Temperature Entropy and the Loss of Bellman Structure}
\label{sec:structural-challenge}

The tractability of the discrete-time no-shorting mean--variance problem relies on a particular homogeneous Bellman structure. This section first recalls that structure and then shows why a standard fixed-temperature entropy regularization fails to preserve it under the no-shorting constraint.

\subsection{The Baseline No-Shorting Recursion}
\label{subsec:original-structure}

Consider a discrete-time market with one riskless asset and $n$ risky assets over periods
$t=0,1,\ldots,T-1$. Let $s_t>0$ denote the gross return of the riskless asset
and $P_t\in\R^n$ the vector of risky excess returns. The random vectors $P_0,\ldots,P_{T-1}$ are assumed to be mutually independent, although they need not be identically distributed. Let $x_t$ denote the investor's wealth and
$u_t\in\R_+^n$ the vector of risky-asset positions, where $\R_+^n$ denotes
the nonnegative orthant of $\R^n$. The decision $u_t$ is made using the
information available at date $t$, before $P_t$ is realized. Wealth evolves
according to
\[
    x_{t+1}=s_tx_t+P_t^\top u_t.
\]
where $^\top$ denotes the transpose of a vector or a matrix. 

Throughout the paper, we assume that, for every
$t=0,1,\ldots,T-1$,
\[
    \E[P_t]>0
    \quad\text{componentwise},
    \qquad
    \E[P_tP_t^\top]\succ0,
    \qquad
    \Var(P_t)\succ0.
\]
The positive-mean condition is the standard positive-risk-premium assumption
used in the no-shorting formulation of \citet{cui2014optimal}. The
pre-commitment multi-period mean--variance problem is
\begin{equation}
\begin{aligned}
    \min_{\{u_t\}_{t=0}^{T-1}}
    \quad & \Var_{0,x_0}(x_T) \\
    \text{s.t.}\quad
    & \E_{0,x_0}[x_T]=d,\\
    & x_{t+1}=s_tx_t+P_t^\top u_t,\\
    & u_t\in\R_+^n,
    \qquad t=0,1,\ldots,T-1.
\end{aligned}
\label{eq:original-mv}
\end{equation}
Here $\E_{0,x_0}[\cdot]$ and $\Var_{0,x_0}(\cdot)$ denote expectation
and variance, respectively, conditional on the initial state
$x_0$ at date $0$, with the randomness generated by future market returns.
Unless explicitly conditioned on a state, expectations involving only
$P_t$, such as $\E[P_t]$ and $\E[P_tP_t^\top]$, and $\Var(P_t)$ are
unconditional.

The usual embedding converts the terminal-mean constraint in
\eqref{eq:original-mv} into a family of quadratic control problems indexed
by a multiplier $\lambda$. Specifically, after introducing the multiplier
for the constraint $\E_{0,x_0}[x_T]=d$ and completing the square, the
Lagrangian can be written, up to the additive term
$-\frac12\lambda^2$ that is independent of the control, as
\[
    \E_{0,x_0}\!\left[
        \frac12\bigl(x_T-(d-\lambda)\bigr)^2
    \right].
\]
Hence, for each fixed $\lambda$, the corresponding auxiliary problem is
\begin{equation}
\begin{aligned}
    \min_{\{u_t\}_{t=0}^{T-1}}
    \quad &
    \E_{0,y_0}\!\left[
        \frac12y_T^2
    \right] \\
    \text{s.t.}\quad
    &y_{t+1}=s_ty_t-P_t^\top u_t,\\
    &u_t\in\R_+^n,
\end{aligned}
\label{eq:auxiliary-mv}
\end{equation}
where
\[
    y_t:=(d-\lambda)q_t^{-1}-x_t,
    \qquad
    q_t:=\prod_{\ell=t}^{T-1}s_\ell,
    \qquad
    q_T=1.
\]
The original constrained problem is recovered by subsequently optimizing
over $\lambda$, with the omitted quadratic term in the multiplier restored.
This separation between the inner dynamic control problem and the outer
optimization over $\lambda$ is also the basis of the regularized formulation
developed below. In particular, all entropy regularization in this paper is
introduced into the auxiliary problem \eqref{eq:auxiliary-mv}, while the
outer optimization over $\lambda$ is retained and solved in
Section~\ref{subsec:optimal-policy-embedding}.

The auxiliary problem has terminal value
\[
    J_T(y_T)=\frac12y_T^2
\]
and Bellman recursion
\[
    J_t(y_t)
    =
    \inf_{u\in\R_+^n}
    \E_{t,y_t}\!\left[
        J_{t+1}\bigl(s_ty_t-P_t^\top u\bigr)
    \right].
\]
Here and below, $\E_{t,y}[\cdot]$ denotes conditional expectation given
$y_t$ at date $t$.  The no-shorting problem admits a finite-dimensional recursion
(\citealp{cui2014optimal}). Suppose that, at date $t+1$,
\[
    J_{t+1}(y)
    =
    \frac12 C_{t+1}q_{t+1}^2y^2\1_{\{y>0\}}
    +
    \frac12 q_{t+1}^2y^2\1_{\{y\le0\}},
\]
where $\1_{\{\cdot\}}$ is the indicator function. 
For $y_t>0$, write
\[
    u_t=s_ty_tK,
    \qquad K\in\R_+^n.
\]
Then
\[
    y_{t+1}=s_ty_t(1-P_t^\top K),
\]
and, using $q_t=s_tq_{t+1}$,
\[
    \E_{t,y_t}[J_{t+1}(y_{t+1})]
    =
    q_t^2y_t^2 h_t(K;C_{t+1}),
\]
where the stage function is 
\begin{equation}
    h_t(K;c)
    :=
    \frac12
    \E\!\left[
        c(1-P_t^\top K)^2\1_{\{P_t^\top K<1\}}
        +
        (1-P_t^\top K)^2\1_{\{P_t^\top K\ge1\}}
    \right],
    \qquad
    K\in\R_+^n.
    \label{eq:stage-energy-new}
\end{equation}
Hence, for $y_t>0$,
\[
    J_t(y_t)
    =
    q_t^2y_t^2
    \min_{K\in\R_+^n}
    h_t(K;C_{t+1}).
\]
For $y_t\le0$, setting $u_t=0$ gives $y_{t+1}=s_ty_t\le0$. Moreover,
the gradient of the Bellman objective at $u_t=0$ is
\[
    \left.
    \nabla_{u_t}
    \E_{t,y_t}\!\left[
        J_{t+1}(s_ty_t-P_t^\top u_t)
    \right]
    \right|_{u_t=0}
    =
    -q_{t+1}^2s_ty_t\,\E[P_t]
    \ge0
\]
componentwise, where the inequality follows from $y_t\le0$ and
$\E[P_t]>0$. Since the Bellman objective is convex over
$u_t\in\R_+^n$, the first-order optimality condition is satisfied at
$u_t=0$; hence the optimal risky position is zero. Therefore,
\[
    J_t(y_t)=\frac12q_t^2y_t^2.
\]
Therefore, starting from $C_T=1$, backward induction gives
\begin{equation}
    J_t(y)
    =
    \frac12 C_tq_t^2y^2\1_{\{y>0\}}
    +
    \frac12 q_t^2y^2\1_{\{y\le0\}},
    \label{eq:original-piecewise-value-new}
\end{equation}
where
\begin{equation}
    K_t^*
    =
    \argmin_{K\in\R_+^n}h_t(K;C_{t+1}),
    \qquad
    C_t
    =
    2\min_{K\in\R_+^n}h_t(K;C_{t+1}),
    \qquad
    C_T=1.
    \label{eq:original-recursion-new}
\end{equation}
The optimal control of the auxiliary problem is
\[
    u_t^*(y_t)
    =
    s_ty_tK_t^*\1_{\{y_t>0\}}.
\]
Thus the auxiliary value function is piecewise quadratic, while the
backward recursion is reduced to the scalar coefficient $C_t$ and the
vector $K_t^*$ at each date.

\subsection{Fixed Temperature: Convexity without Bellman Closure}
\label{subsec:fixed-temperature-challenge}

We first regularize the fixed-$\lambda$ auxiliary problem
\eqref{eq:auxiliary-mv} using a fixed absolute entropy temperature
$\bar\tau>0$. Conditional on the transformed state $y_t$, the risky
allocation $U_t$ is randomized according to a probability density
$\rho_t(\cdot\mid y_t)$ with respect to Lebesgue measure, supported on
$\R_+^n$. Thus
\[
    \rho_t(u\mid y_t)\ge0,
    \qquad
    \int_{\R_+^n}\rho_t(u\mid y_t)\,du=1,
\]
and every realization of $U_t$ satisfies the no-shorting constraint.
The transformed state evolves according to
\[
    y_{t+1}=s_ty_t-P_t^\top U_t.
\]

For a given initial transformed state
\(
    y_0,
\)
the fixed-temperature entropy-regularized auxiliary problem is
\begin{equation}
\begin{aligned}
    \min_{\{\rho_t\}_{t=0}^{T-1}}\quad &
    \E_{0,y_0}\!\left[
        \frac12y_T^2
        +
        \bar\tau
        \sum_{t=0}^{T-1}
         \int_{\R_+^n}
    \rho_t(u\mid y_t) \log\rho_t(u\mid y_t) du
    \right],\\
    \text{s.t.}\quad & y_{t+1} = s_t y_t- P_t^\top u_t,\\
    & u_t \sim \rho_t(u\mid y_t),\quad     \rho_t(u\mid y_t)\ge0,
    \quad
    \int_{\R_+^n}\rho_t(u\mid y_t)\,du=1,
    \end{aligned}
    \label{eq:fixed-temperature-objective}
\end{equation}
where the infimum is taken over admissible Markov densities supported
on $\R_+^n$ for which the entropy terms are well defined, and the
expectation is taken over both market and policy randomness. 

For an arbitrary date $t$ and state $y_t$, define the continuation value
\[
    \widetilde J_t^{\bar\tau}(y_t)
    :=
    \min_{\{\rho_s\}_{s=t}^{T-1}}
    \E_{t,y_t}\!\left[
        \frac12y_T^2
        +
        \bar\tau
        \sum_{s=t}^{T-1}
           \int_{\R_+^n}
    \rho_t(u\mid y_t) \log\rho_t(u\mid y_t) du
    \right],
\]
with terminal condition
\[
    \widetilde J_T^{\bar\tau}(y_T)
    =
    \frac12y_T^2.
\]
The Bellman principle therefore gives
\begin{equation}
\begin{aligned}
    \widetilde J_t^{\bar\tau}(y_t)
    =
    \min_{\rho_t}
    \bigg\{
        &\int_{\R_+^n}
        \E_{t,y_t}\!\left[
            \widetilde J_{t+1}^{\bar\tau}
            (s_ty_t-P_t^\top u)
        \right]
        \rho_t(u\mid y)\,du+
        \bar\tau
        \int_{\R_+^n}
        \rho_t(u\mid y)\log\rho_t(u\mid y)\,du
    \bigg\},
\end{aligned}
\label{eq:fixed-temperature-bellman}
\end{equation}
where the infimum is over probability densities on $\R_+^n$ for which
the objective is finite. Whenever the normalizing integral is finite, the entropy variational
identity yields the unique optimal density
\begin{equation}
    \widetilde\rho_t^{\bar\tau}(u\mid y_t)
    =
    \frac{
        \exp\!\left\{
            -\frac{
                \E_{t,y_t}\!\left[
                    \widetilde J_{t+1}^{\bar\tau}
                    (s_ty_t-P_t^\top u)
                \right]
            }{\bar\tau}
        \right\}
    }{
        \int_{\R_+^n}
        \exp\!\left\{
            -\frac{
                \E_{t,y_t}\!\left[
                    \widetilde J_{t+1}^{\bar\tau}
                    (s_ty_t-P_t^\top \nu)
                \right]
            }{\bar\tau}
        \right\}d\nu
    } ,
    \label{eq:fixed-temperature-gibbs}
\end{equation}
and the corresponding soft Bellman value is
\begin{equation}
    \widetilde J_t^{\bar\tau}(y_t)
    =
    -\bar\tau
    \log
    \int_{\R_+^n}
    \exp\!\left\{
        -\frac{
            \E_{t,y_t}\!\left[
                \widetilde J_{t+1}^{\bar\tau}
                (s_ty_t-P_t^\top u)
            \right]
        }{\bar\tau}
    \right\}du .
    \label{eq:fixed-temperature-softmin}
\end{equation}

The fixed-temperature formulation therefore defines a well-posed soft
Bellman recursion whenever the corresponding integrals are finite. The
following result shows that regularizing the auxiliary problem in this
way preserves convexity but generally destroys the piecewise-quadratic
Bellman structure inherited from the original no-shorting problem.

\begin{proposition}
\label{prop:fixed-temperature-benchmark-new}
Suppose that the integrals in
\eqref{eq:fixed-temperature-softmin} are finite.
Then:

\begin{enumerate}
    \item
    $\widetilde J_t^{\bar\tau}(y_t)$ is convex in the state $y_t$ for every
    $t=0,1,\ldots,T$.

    \item
    The two-branch piecewise-quadratic value-function structure of the original
no-shorting auxiliary problem is not, in general, preserved under
fixed-temperature entropy regularization. In particular, if
    \[
        \bigl(\E[P_{T-1}P_{T-1}^{\top}]\bigr)^{-1}
        \E[P_{T-1}]
        \in\R_+^n\setminus\{0\},
    \]
    then $\widetilde J_{T-1}^{\bar\tau}(y_{T-1})$ is not quadratic on the region of $y_{T-1}>0$.
\end{enumerate}
\end{proposition}

\begin{proof}
We first prove convexity by backward induction. At the terminal date,
\[
    \widetilde J_T^{\bar\tau}(y_T)=\frac12y_T^2
\]
is convex. Suppose that
$\widetilde J_{t+1}^{\bar\tau}(y_{t+1})$ is convex. Since
\[
    (y_t,u)\longmapsto s_ty_t-P_t^\top u
\]
is affine,
\[
    (y_t,u)
    \longmapsto
    \E_{t,y_t}\!\left[
        \widetilde J_{t+1}^{\bar\tau}
        (s_ty_t-P_t^\top u)
    \right]
\]
is jointly convex. Therefore
\[
    (y_t,u)
    \longmapsto
    \exp\!\left\{
        -\frac{
        \E_{t,y_t}\!\left[
            \widetilde J_{t+1}^{\bar\tau}
            (s_ty_t-P_t^\top u)
        \right]
        }{\bar\tau}
    \right\}
\]
is jointly log-concave. Since $\R_+^n$ is convex, integration with respect to $u$ preserves log-concavity. It follows from
\eqref{eq:fixed-temperature-softmin} that
$\widetilde J_t^{\bar\tau}(y_t)$ is convex. This proves the first statement by induction.

To establish the second statement, consider the final decision date $T-1$. We have 
\[
    \widetilde J_{T-1}^{\bar\tau}(y_{T-1})
    =
    -\bar\tau
    \log
    \int_{\R_+^n}
    \exp\!\left\{
        -\frac{
        \E_{T-1,y_{T-1}}\!\left[
        (s_{T-1}y_{T-1}-P_{T-1}^{\top}u)^2
        \right]
        }{2\bar\tau}
    \right\}du.
\]
Expanding the quadratic term of $u$,
\begin{align*}
    &\E_{T-1,y_{T-1}}\!\left[
        (s_{T-1}y_{T-1}-P_{T-1}^{\top}u)^2
    \right]
    \\
    =&
    u^\top \E[P_{T-1}P_{T-1}^\top]u
    -2s_{T-1}y_{T-1}\E[P_{T-1}]^\top u
    +s_{T-1}^2y_{T-1}^2
    \\
    =&
    \bigl(
        u
        -
        s_{T-1}y_{T-1}
        \bigl(\E[P_{T-1}P_{T-1}^\top]\bigr)^{-1}
        \E[P_{T-1}]
    \bigr)^\top
    \E[P_{T-1}P_{T-1}^\top]
    \bigl(
        u
        -
        s_{T-1}y_{T-1}
        \bigl(\E[P_{T-1}P_{T-1}^\top]\bigr)^{-1}
        \E[P_{T-1}]
    \bigr)
    \\
    &
    +
    s_{T-1}^2y_{T-1}^2
    \left[
        1
        -
        \E[P_{T-1}]^\top
        \bigl(\E[P_{T-1}P_{T-1}^\top]\bigr)^{-1}
        \E[P_{T-1}]
    \right].
\end{align*}
Substituting this expression into the integral gives
\begin{align*}
    \widetilde J_{T-1}^{\bar\tau}(y_{T-1})
    ={}&
    \frac12s_{T-1}^2y_{T-1}^2
    \left[
        1
        -
        \E[P_{T-1}]^\top
        \bigl(\E[P_{T-1}P_{T-1}^\top]\bigr)^{-1}
        \E[P_{T-1}]
    \right] -\frac{n\bar\tau}{2}\log(2\pi\bar\tau) \\
    & 
    +\frac{\bar\tau}{2}
    \log\det\left(\E[P_{T-1}P_{T-1}^\top]\right)
 -\bar\tau
    \log
    \Prob\!\left\{
       Z_{T-1}
        \in\R_+^n
    \right\},
\end{align*}
where $\det(\cdot)$ denotes the determinant of a matrix,
$\Prob\{\cdot\}$ denotes the probability of an event, and
$Z_{T-1}\sim N\!\left(
s_{T-1}y_{T-1}
(\E[P_{T-1}P_{T-1}^\top])^{-1}\E[P_{T-1}],
\bar\tau(\E[P_{T-1}P_{T-1}^\top])^{-1}
\right)$.
The first three terms are quadratic in $y_{T-1}$ or independent of $y_{T-1}$. Thus the only term that can destroy quadraticity is the logarithm of the Gaussian orthant probability.

Under the assumption of 
\[
    \bigl(\E[P_{T-1}P_{T-1}^{\top}]\bigr)^{-1}
    \E[P_{T-1}]
    \in\R_+^n\setminus\{0\},
\]
the mean of the Gaussian distribution moves componentwise in a nonnegative
direction as $y_{T-1}$ increases on the region of $y_{T-1}>0$, and strictly in at least one component,
while its covariance matrix remains fixed and positive definite. Hence the
probability of the nonnegative orthant is strictly increasing in
$y_{T-1}$ on $[0,\infty)$. At $y_{T-1}=0$, the Gaussian distribution is nondegenerate and has a
strictly positive density on $\R^n$. Since $\R_+^n$ has positive Lebesgue
measure,
\[
    \Prob\!\left\{
        Z
        \in\R_+^n
    \right\}>0,
\]
where $Z\sim N\!\left(
            0,
            \bar\tau
            (\E[P_{T-1}P_{T-1}^{\top}])^{-1}
        \right).$
Therefore, for all $y_{T-1}\ge0$, the orthant probability is bounded below
by a strictly positive constant and above by one. Its logarithm is thus
bounded and nonconstant on $[0,\infty)$.

A nonconstant quadratic polynomial cannot be bounded on $[0,\infty)$. Hence the logarithmic probability term is not quadratic in $y_{T-1}$, and consequently
$\widetilde J_{T-1}^{\bar\tau}(y_{T-1})$ is not quadratic on $y_{T-1}>0$. Since the original no-shorting problem requires a quadratic positive-state branch at every date, the fixed-temperature entropy regularization does not preserve that family.
\end{proof}

A fixed absolute temperature therefore preserves convexity at the cost of losing the low-dimensional recursion of the original no-shorting problem. In the original model without entropy regularization, the value function is characterized by the scalar continuation coefficient $C_t$, while the portfolio rule is determined by the coefficient $K_t^*$. The wealth state enters only through the factor $q_t^2y_t^2$. Once the piecewise-quadratic structure is lost, the continuation value must instead be treated as a general function
$\widetilde J_t^{\bar\tau}(y_t)$. Each Bellman recursion then requires evaluating this function at the random next state and performing the corresponding expectation and soft integration over the feasible control set. The backward recursion is therefore no longer a sequence of low-dimensional coefficient problems, but a function-valued dynamic program. This computational cost is particularly severe in the sample-based setting. Since the expectation is estimated from return observations, losing the scalar coefficient recursion means that sample averages must be recomputed over many state values and candidate controls at every date.

The fixed temperature also destroys the scale separation of the original feedback rule. Under no-shorting,
\[
    u_t^*
    =
    s_ty_tK_t^*\1_{\{y_t>0\}},
\]
so the wealth gap $y_t$ determines the scale of risky exposure and the state-independent vector $K_t^*$ determines its composition. With a fixed absolute temperature, the mean--variance loss scales with $y_t^2$ whereas the entropy penalty does not. The relative strength of regularization therefore varies with the wealth gap, and the optimal randomized policy cannot in general be obtained by scaling a state-independent distribution over portfolio coefficients. The piecewise-linear wealth-gap interpretation of the original policy is consequently lost.

These two limitations motivate a state-dependent temperature that preserves the homogeneity of the original mean--variance recursion.

\section{State-Dependent Entropy and Bellman Structure Preservation}
\label{sec:structure-preserving-soft}

Section~\ref{sec:structural-challenge} shows that applying a fixed absolute entropy temperature to
the embedded auxiliary problem (\ref{eq:auxiliary-mv}) preserves convexity but destroys
its homogeneous piecewise-quadratic structure. We now construct a
state-dependent entropy regularization of the same auxiliary problem
that restores this structure. The construction leads to a low-dimensional soft Bellman recursion, but convexity is no longer automatic because entropy enters the continuation coefficients. We therefore derive a whole-horizon condition that keeps the recursion in the positive-coefficient region and provide a directly checkable sufficient bound. The section then characterizes the optimal randomized policy, the corresponding mode feedback and embedding multiplier, and their zero-temperature limits.

\subsection{Structure Preservation and Whole-Horizon Convexity}
\label{subsec:state-dependent-convexity}

We continue to work with the fixed-$\lambda$ auxiliary problem
\eqref{eq:auxiliary-mv}. The solution of the auxiliary problem  has two
branches. When $y_t\le0$, the optimal risky position is zero under the
standing positive-mean condition, so no entropy regularization is needed on
this branch. When $y_t>0$, the risky allocation takes the scale-separated
form $u_t=s_ty_tK$ with $K\in\R_+^n$. We therefore introduce entropy
regularization through a probability density $\pi_t(K)$ over the coefficient $K$ on
the positive-state branch, while retaining the original inactive branch for
$y_t\le0$.

Let $\tau_t(y_t)>0$ denote the entropy coefficient on the positive-state
branch. For a given initial transformed state $y_0$, the corresponding
state-dependent entropy-regularized auxiliary problem is
\begin{equation}
\begin{aligned}
    \min_{\{\pi_t\}_{t=0}^{T-1}}\quad &
    \E_{0,y_0}\Bigg[
        \frac12y_T^2
        +
        \sum_{t=0}^{T-1}
        \tau_t(y_t)\1_{\{y_t>0\}}
        \int_{\R_+^n}
        \pi_t(K)
        \log\pi_t(K)\,dK
    \Bigg].\\
    \text{s.t.}\quad & y_{t+1}= s_ty_t-s_t y_tP_t^\top K_t \1_{\{y_t>0\}},\\
    & K_t \sim \pi_t(K),\quad \pi(K)\ge 0, \quad \int_{\R_+^n}  \pi_t(K) dK = 1,
    \label{eq:state-dependent-auxiliary-general}
\end{aligned}
\end{equation}
The expectation is taken over both market and policy randomness. The Bellman principle therefore gives
\begin{equation}
\begin{aligned}
    J_t^{\tau}(y_t)
    =
    \min_{\pi_t}
    \bigg\{
        &\int_{\R_+^n}
        \E_{t,y_t}\!\left[
            J_{t+1}^{\tau}
            (y_{t+1})
        \right]
        \pi_t(K)\,dK+
        \tau_t(y_t)\1_{\{y_t>0\}}
        \int_{\R_+^n}
        \pi_t(K)\log\pi_t(K)\,dK
    \bigg\},
\end{aligned}
\label{eq:state-dependent-temperature-bellman}
\end{equation}
where $J_{t+1}^{\tau}
            (y_{t+1})$ is the continuation value function at date $t+1$. 

To preserve the low-dimensional recursion of the original auxiliary problem 
problem, the entropy regularization must be compatible with the homogeneity
of the embedded Bellman equation. The following result characterizes the
required temperature scaling and states the conditions under which the
piecewise-quadratic structure is retained.

\begin{proposition}
\label{prop:structure-compatible-scaling}
Suppose that the continuation value at date $t+1$ has the form
\[
    J_{t+1}^\tau(y_{t+1})
    =
    \frac12 C_{t+1}^\tau q_{t+1}^2y_{t+1}^2\1_{\{y_{t+1}>0\}}
    +
    \frac12 q_{t+1}^2y_{t+1}^2\1_{\{y_{t+1}\le0\}},
\]
with $C_{t+1}^\tau>0$. Under the state-dependent temperature
\begin{equation}
    \tau_t(y_t)
    =
    q_t^2y_t^2\tau,
    \qquad \tau>0,
    \label{eq:state-dependent-temperature-new}
\end{equation}
the positive-state Bellman objective can be written as
$q_t^2y_t^2$ times a variational problem that is independent of $y_t$.

Whenever
\(
    \int_{\R_+^n}
    \exp\!\left\{
        -\frac{h_t(K;C_{t+1}^\tau)}{\tau}
    \right\}dK
    <\infty,
\)
define
\begin{equation}
    C_t^\tau
    :=
    -2\tau
    \log
    \int_{\R_+^n}
    \exp\!\left\{
        -\frac{h_t(K;C_{t+1}^\tau)}{\tau}
    \right\}dK.
    \label{eq:soft-logpartition-new}
\end{equation}
If $C_t^\tau>0$, then
\[
    J_t^\tau(y)
    =
    \frac12 C_t^\tau q_t^2y^2\1_{\{y>0\}}
    +
    \frac12 q_t^2y^2\1_{\{y\le0\}},
\]
and $J_t^\tau(y_t)$ is convex in $y_t$.
\end{proposition}

\begin{proof}
For $y_t>0$, using the assumed form of $J_{t+1}^\tau(y_{t+1})$ and
$q_t=s_tq_{t+1}$ gives
\[
    \E_{t,y_t}\!\left[
        J_{t+1}^\tau(y_{t+1})
    \right]
    =
    q_t^2y_t^2
    h_t(K;C_{t+1}^\tau).
\]
Then, the objective in Bellman equation (\ref{eq:state-dependent-temperature-bellman}) becomes
\[
    q_t^2y_t^2
    \int_{\R_+^n}
    h_t(K;C_{t+1}^\tau)\pi_t(K)\,dK
    +
    \tau_t(y_t)
    \int_{\R_+^n}
    \pi_t(K)\log\pi_t(K)\,dK.
\]
Under \eqref{eq:state-dependent-temperature-new}, this expression can be
written as
\[
    q_t^2y_t^2
    \left[
        \int_{\R_+^n}
        h_t(K;C_{t+1}^\tau)\pi_t(K)\,dK
        +
        \tau
        \int_{\R_+^n}
        \pi_t(K)\log\pi_t(K)\,dK
    \right].
\]
Hence the state-dependent factor $q_t^2y_t^2$ separates from the
optimization over $\pi_t$, and the remaining variational problem is
independent of $y_t$.

When the integral in (\ref{eq:soft-logpartition-new})
 is finite, the
entropy variational identity yields
\[
\begin{aligned}
    &\min_{\pi_t}
    \left\{
        \int_{\R_+^n}
        h_t(K;C_{t+1}^\tau)\pi_t(K)\,dK
        +
        \tau
        \int_{\R_+^n}
        \pi_t(K)\log\pi_t(K)\,dK
    \right\}
    \\
   = &
    -\tau
    \log
    \int_{\R_+^n}
    \exp\!\left\{
        -\frac{h_t(K;C_{t+1}^\tau)}{\tau}
    \right\}dK
    =
    \frac12 C_t^\tau.
\end{aligned}
\]
Therefore, on the positive-state branch,
\[
    J_t^\tau(y_t)
    =
    \frac12 C_t^\tau q_t^2y_t^2,
    \qquad
    y_t>0.
\]

For $y_t\le0$, no randomized control is introduced
on this branch and the 
optimal risky position remains zero following the argument of \cite{cui2014optimal}. Thus, 
\[
    J_t^\tau(y_t)
    =
    \frac12 q_t^2y_t^2.
\]
Combining the two branches gives
\[
    J_t^\tau(y_t)
    =
    \frac12 C_t^\tau q_t^2y_t^2\1_{\{y_t>0\}}
    +
    \frac12 q_t^2y_t^2\1_{\{y_t\le0\}}.
\]

Finally, if $C_t^\tau>0$, both quadratic branches are convex. They agree
at $y=0$ and have the same derivative there, equal to zero. Hence
$J_t^\tau(y_t)$ is convex on $\R$.
\end{proof}

Consequently, if the continuation coefficients remain positive throughout
the backward recursion, both the original piecewise-quadratic
value-function structure and convexity are preserved over the entire
horizon. Economically, the scaling $\tau_t(y_t)=q_t^2y_t^2\tau$ makes the strength of entropy regularization proportional to the squared wealth gap that governs both risky exposure and the mean--variance continuation loss, so its relative importance remains stable across positive states and vanishes as the wealth gap approaches zero.

The remaining issue is whether the continuation coefficients can be kept
strictly positive throughout the backward recursion. A negative
$C_t^\tau$ makes the positive-state branch of $J_t^\tau(y_t)$ concave, while
positivity of $C_{t}^\tau$ ensures that the preceding stage function 
$h_{t-1}(\cdot;C_{t}^\tau)$ remains strictly convex and coercive under the
standing moment assumptions. We therefore seek a common temperature
range that preserves positivity over the entire horizon. 

\begin{theorem}
\label{thm:whole-horizon-convexity-new}
Let $\{C_t \}_{t=0}^T$ be the continuation coefficients of the original
no-shorting recursion in \eqref{eq:original-recursion-new}, with $C_t>0 $. Then there
exists $\tau_0>0$ such that, for every $0<\tau<\tau_0$:

\begin{enumerate}
    \item for every $t=0,1,\ldots,T-1$,
    \[
        \int_{\R_+^n}
        \exp\!\left\{
            -\frac{h_t(K;C_{t+1}^\tau)}{\tau}
        \right\}dK
        <\infty;
    \]

    \item
    \[
        C_t^\tau>0,
        \qquad
        t=0,1,\ldots,T;
    \]

    \item $J_t^\tau(\cdot)$ is convex on $\R$ for every
    $t=0,1,\ldots,T$;

    \item $h_t(\cdot;C_{t+1}^\tau)$ is strictly convex and coercive on
    $\R_+^n$ for every $t=0,1,\ldots,T-1$.
\end{enumerate}

Hence the structure-preserving auxiliary Bellman recursion remains a
sequence of well-posed convex problems over the entire horizon. Moreover,
\begin{equation}
    C_t^\tau\longrightarrow C_t,
    \qquad
    \tau\downarrow0,
    \qquad
    t=0,1,\ldots,T.
    \label{eq:continuation-recovery-new}
\end{equation}
\end{theorem}

\begin{proof}
We proceed by backward induction. At $t=T$,
\[
    C_T^\tau=C_T=1,
\]
so the result is immediate.

Suppose that, for some $t<T$, $C_{t+1}^\tau$ is well defined for all
sufficiently small $\tau$ and
\[
    C_{t+1}^\tau\longrightarrow C_{t+1}>0,
    \qquad
    \tau\downarrow0.
\]
Then, for all sufficiently small $\tau$,
\[
    C_{t+1}^\tau
    \ge
    \frac12C_{t+1}>0.
\]

We first establish the curvature properties of the stage problem. For a
fixed realization of $P_t$, the integrand in
$h_t(K;C_{t+1}^\tau)$ is
\[
    \frac12C_{t+1}^\tau(1-P_t^\top K)^2
    \qquad\text{if }P_t^\top K<1,
\]
and
\[
    \frac12(1-P_t^\top K)^2
    \qquad\text{if }P_t^\top K\ge1.
\]
The two pieces agree in value and first derivative at
$P_t^\top K=1$. Since $C_{t+1}^\tau>0$, both pieces are convex.
Moreover, for every nonzero direction $v$,
the expected second directional curvature is bounded below by
\[
    \min\{C_{t+1}^\tau,1\}
    v^\top\E[P_tP_t^\top]v>0.
\]
Hence $h_t(\cdot;C_{t+1}^\tau)$ is strictly convex. In addition,
\[
    h_t(K;C_{t+1}^\tau)
    \ge
    \frac12
    \min\!\left\{
        \frac12C_{t+1},1
    \right\}
    \E[(1-P_t^\top K)^2].
\]
Because $\E[P_tP_t^\top]\succ0$, the right-hand side is a coercive
quadratic function of $K$. Thus there exist constants $a>0$ and
$b<\infty$, independent of sufficiently small $\tau$, such that
\[
    h_t(K;C_{t+1}^\tau)
    \ge
    a\|K\|^2-b,
    \qquad
    K\in\R_+^n.
\]
It follows that
\[
    \int_{\R_+^n}
    \exp\!\left\{
        -\frac{h_t(K;C_{t+1}^\tau)}{\tau}
    \right\}dK
    <\infty.
\]

Next, from the definition of $h_t$,
\[
\begin{aligned}
    &h_t(K;C_{t+1}^\tau)
    -
    h_t(K;C_{t+1})
=
    \frac12
    (C_{t+1}^\tau-C_{t+1})
    \E\!\left[
        (1-P_t^\top K)^2
        \1_{\{P_t^\top K<1\}}
    \right].
\end{aligned}
\]
Hence
\[
    h_t(K;C_{t+1}^\tau)
    \longrightarrow
    h_t(K;C_{t+1})
\]
uniformly on every compact subset of $\R_+^n$. Together with the common
coercive lower bound, this keeps the minimizers in a common compact set
and gives
\begin{equation}
    \min_{K\in\R_+^n}
    h_t(K;C_{t+1}^\tau)
    \longrightarrow
    \min_{K\in\R_+^n}
    h_t(K;C_{t+1}).
    \label{eq:min-energy-convergence}
\end{equation}

It remains to establish the corresponding limit of the soft minimum.
Since $C_{t+1}^\tau$ is uniformly bounded away from zero, the preceding
curvature argument gives a constant $\gamma>0$, independent of
sufficiently small $\tau$, such that
$h_t(\cdot;C_{t+1}^\tau)$ is $\gamma$-strongly convex. Let
$K_t^{\tau,\sharp}$ be its unique minimizer. Then
\[
    h_t(K;C_{t+1}^\tau)
    \ge
    \min_{\nu\in\R_+^n}
    h_t(\nu;C_{t+1}^\tau)
    +
    \frac{\gamma}{2}
    \|K-K_t^{\tau,\sharp}\|^2.
\]
Consequently,
\begin{align*}
    &\int_{\R_+^n}
    \exp\!\left\{
        -\frac{h_t(K;C_{t+1}^\tau)}{\tau}
    \right\}dK
\le
    \exp\!\left\{
        -\frac{
            \min_{\nu\in\R_+^n}
            h_t(\nu;C_{t+1}^\tau)
        }{\tau}
    \right\}
    \left(
        \frac{2\pi\tau}{\gamma}
    \right)^{n/2}.
\end{align*}
Therefore,
\[
\begin{aligned}
    &-\tau
    \log
    \int_{\R_+^n}
    \exp\!\left\{
        -\frac{h_t(K;C_{t+1}^\tau)}{\tau}
    \right\}dK
\ge
    \min_{K\in\R_+^n}
    h_t(K;C_{t+1}^\tau)
    -
    \frac{n\tau}{2}
    \log\!\left(
        \frac{2\pi\tau}{\gamma}
    \right).
\end{aligned}
\]
Using \eqref{eq:min-energy-convergence} and $\tau\log\tau\to0$ yields
\[
\begin{aligned}
    \liminf_{\tau\downarrow0}
    \left[
        -\tau
        \log
        \int_{\R_+^n}
        \exp\!\left\{
            -\frac{h_t(K;C_{t+1}^\tau)}{\tau}
        \right\}dK
    \right]
    \ge
    \min_{K\in\R_+^n}
    h_t(K;C_{t+1}).
\end{aligned}
\]

For the reverse inequality, let $K_t^*$ be the minimizer of
$h_t(\cdot;C_{t+1})$. For any $\varepsilon>0$, continuity and local
uniform convergence imply that there exists $r>0$ such that, for all
sufficiently small $\tau$,
\[
    h_t(K;C_{t+1}^\tau)
    \le
    \min_{\nu\in\R_+^n}
    h_t(\nu;C_{t+1})
    +
    \varepsilon
\]
for every
\[
    K\in B_r(K_t^*)\cap\R_+^n,
    \qquad
    B_r(K_t^*)
    :=
    \{K\in\R^n:\|K-K_t^*\|<r\}.
\]
Since this set has positive Lebesgue measure,
\begin{align*}
    &\int_{\R_+^n}
    \exp\!\left\{
        -\frac{h_t(K;C_{t+1}^\tau)}{\tau}
    \right\}dK
\ge
    \left|
        B_r(K_t^*)\cap\R_+^n
    \right|
    \exp\!\left\{
        -\frac{
            \min_{\nu\in\R_+^n}
            h_t(\nu;C_{t+1})
            +
            \varepsilon
        }{\tau}
    \right\}.
\end{align*}
Hence
\[
\begin{aligned}
    \limsup_{\tau\downarrow0}
    \left[
        -\tau
        \log
        \int_{\R_+^n}
        \exp\!\left\{
            -\frac{h_t(K;C_{t+1}^\tau)}{\tau}
        \right\}dK
    \right]
    \le
    \min_{K\in\R_+^n}
    h_t(K;C_{t+1})
    +
    \varepsilon.
\end{aligned}
\]
Letting $\varepsilon\downarrow0$ gives
\[
    -\tau
    \log
    \int_{\R_+^n}
    \exp\!\left\{
        -\frac{h_t(K;C_{t+1}^\tau)}{\tau}
    \right\}dK
    \longrightarrow
    \min_{K\in\R_+^n}
    h_t(K;C_{t+1}).
\]

Using \eqref{eq:soft-logpartition-new},
\[
    C_t^\tau
    \longrightarrow
    2\min_{K\in\R_+^n}
    h_t(K;C_{t+1})
    =
    C_t>0.
\]
Thus $C_t^\tau>0$ for all sufficiently small $\tau$.
Proposition~\ref{prop:structure-compatible-scaling} then gives the
piecewise-quadratic and convex value function, while the curvature
argument above gives strict convexity and coercivity of
$h_t(\cdot;C_{t+1}^\tau)$. This completes the induction.

Finally, because the horizon is finite, the minimum of the finitely many
stagewise admissible temperature bounds is strictly positive. Choosing this
minimum as $\tau_0$ gives one common temperature range over which all
conclusions hold simultaneously. The same backward induction establishes
\eqref{eq:continuation-recovery-new}.
\end{proof}

Theorem~\ref{thm:whole-horizon-convexity-new} is the key
positive-temperature result. The state-dependent entropy scale preserves the
low-dimensional Bellman structure, while a sufficiently small common base
temperature prevents the recursively generated curvature coefficients from
crossing into the nonconvex region. The zero-temperature convergence in
\eqref{eq:continuation-recovery-new} is a consistency property of the
construction; the small-temperature restriction is already needed at a
strictly positive temperature to preserve convexity over the full horizon.
The theorem is existential. A more conservative but directly checkable
sufficient condition can be obtained from the first two moments of returns.

\begin{proposition}
\label{prop:explicit-temperature-bound-new}
Define
\[
    \Delta^*
    :=
    \min_{0\le t\le T-1}
    \left(
        1-
        \E[P_t]^\top
        \E[P_tP_t^\top]^{-1}
        \E[P_t]
    \right),
\]
and
\[
    d^*
    :=
    \min_{0\le t\le T-1}
    \det\!\left(
        \E[P_tP_t^\top]
    \right).
\]
Suppose $0<\Delta^*<1$ and
\[
    d^*
    \ge
    \left(
        \frac{2\pi e(1-\Delta^*)}{n}
    \right)^n.
\]
Then, there exists $\kappa>0$ satisfying
\[
    \kappa(\Delta^*-1)
    -
    n\log\!\left(
        \frac{2\pi}{\kappa}
    \right)
    +
    \log d^*
    \ge0.
\]
For any such $\kappa$, every
\[
    0<\tau<\frac1\kappa
\]
satisfies
\[
    C_t^\tau\ge\kappa\tau>0,
    \qquad
    t=0,1,\ldots,T-1.
\]
\end{proposition}

\begin{proof}
We first show that the stated condition guarantees the existence of
$\kappa>0$ satisfying
\[
    \kappa(\Delta^*-1)
    -
    n\log\!\left(\frac{2\pi}{\kappa}\right)
    +
    \log d^*
    \ge0.
\]
Define
\[
    F(\kappa)
    :=
    \kappa(\Delta^*-1)
    -
    n\log\!\left(\frac{2\pi}{\kappa}\right)
    +
    \log d^*,
    \qquad \kappa>0.
\]
Since $0<\Delta^*<1$,
\[
    F'(\kappa)
    =
    \Delta^*-1+\frac{n}{\kappa},
    \qquad
    F''(\kappa)
    =
    -\frac{n}{\kappa^2}<0.
\]
Thus $F$ is strictly concave and attains its unique maximum at
\[
    \kappa^*
    =
    \frac{n}{1-\Delta^*}.
\]
At this point,
\[
\begin{aligned}
    F(\kappa^*)
    &=
    \log d^*
    +
    n\log\!\left(
        \frac{n}{2\pi e(1-\Delta^*)}
    \right).
\end{aligned}
\]
The assumed condition
\[
    d^*
    \ge
    \left(
        \frac{2\pi e(1-\Delta^*)}{n}
    \right)^n
\]
is therefore equivalent to
\[
    F(\kappa^*)\ge0.
\]
Hence there exists at least one $\kappa>0$ such that
$F(\kappa)\ge0$.

Fix any such $\kappa$ and any
\[
    0<\tau<\frac1\kappa.
\]
We prove
\[
    C_t^\tau\ge\kappa\tau,
    \qquad t=0,1,\ldots,T,
\]
by backward induction. Since $C_T^\tau=1$ and $\kappa\tau<1$, the
inequality holds at $t=T$.

Suppose that, for some $t<T$,
\[
    C_{t+1}^\tau\ge\kappa\tau.
\]
Because $\kappa\tau<1$,
\[
    \min\{C_{t+1}^\tau,1\}
    \ge\kappa\tau.
\]
From the definition of the stage function,
\[
\begin{aligned}
    h_t(K;C_{t+1}^\tau)
    &=
    \frac12
    \E\!\left[
        C_{t+1}^\tau
        (1-P_t^\top K)^2
        \1_{\{P_t^\top K<1\}}
        +
        (1-P_t^\top K)^2
        \1_{\{P_t^\top K\ge1\}}
    \right]
    \\
    &\ge
    \frac12
    \min\{C_{t+1}^\tau,1\}
    \E[(1-P_t^\top K)^2]
    \\
    &\ge
    \frac{\kappa\tau}{2}
    \E[(1-P_t^\top K)^2].
\end{aligned}
\]
It follows that
\[
\begin{aligned}
    &\int_{\R_+^n}
    \exp\!\left\{
        -\frac{h_t(K;C_{t+1}^\tau)}{\tau}
    \right\}dK
\le
    \int_{\R^n}
    \exp\!\left\{
        -\frac{\kappa}{2}
        \E[(1-P_t^\top K)^2]
    \right\}dK.
\end{aligned}
\]

Since $\E[P_tP_t^\top]\succ0$, completing the square gives
\[
\begin{aligned}
    \E[(1-P_t^\top K)^2]
    &=
    \left(
        K-\E[P_tP_t^\top]^{-1}\E[P_t]
    \right)^\top
    \E[P_tP_t^\top]
    \left(
        K-\E[P_tP_t^\top]^{-1}\E[P_t]
    \right)
    \\
    &\quad+
    1-
    \E[P_t]^\top
    \E[P_tP_t^\top]^{-1}
    \E[P_t].
\end{aligned}
\]
Therefore,
\[
\begin{aligned}
    &\int_{\R_+^n}
    \exp\!\left\{
        -\frac{h_t(K;C_{t+1}^\tau)}{\tau}
    \right\}dK
    \\
   \le &
    \exp\!\left\{
        -\frac{\kappa}{2}
        \left[
            1-
            \E[P_t]^\top
            \E[P_tP_t^\top]^{-1}
            \E[P_t]
        \right]
    \right\}
    \left(\frac{2\pi}{\kappa}\right)^{n/2}
    \det\!\left(\E[P_tP_t^\top]\right)^{-1/2}.
\end{aligned}
\]

Using
\[
    C_t^\tau
    =
    -2\tau
    \log
    \int_{\R_+^n}
    \exp\!\left\{
        -\frac{h_t(K;C_{t+1}^\tau)}{\tau}
    \right\}dK,
\]
we obtain
\[
\begin{aligned}
    C_t^\tau
    &\ge
    \tau
    \bigg[
        \kappa
        \left(
            1-
            \E[P_t]^\top
            \E[P_tP_t^\top]^{-1}
            \E[P_t]
        \right)
        -
        n\log\!\left(\frac{2\pi}{\kappa}\right)
        +
        \log\det\!\left(\E[P_tP_t^\top]\right)
    \bigg]
    \\
    &\ge
    \tau
    \left[
        \kappa\Delta^*
        -
        n\log\!\left(\frac{2\pi}{\kappa}\right)
        +
        \log d^*
    \right].
\end{aligned}
\]
Since $F(\kappa)\ge0$,
\[
    \kappa\Delta^*
    -
    n\log\!\left(\frac{2\pi}{\kappa}\right)
    +
    \log d^*
    \ge\kappa,
\]
and hence
\[
    C_t^\tau\ge\kappa\tau.
\]
This closes the backward induction. Therefore,
\[
    C_t^\tau\ge\kappa\tau>0,
    \qquad
    t=0,1,\ldots,T-1.
\]
\end{proof}

\subsection{Optimal policy and embedding solution}
\label{subsec:optimal-policy-embedding}

We now solve the state-dependent entropy-regularized auxiliary problem
for a fixed embedding multiplier and then optimize over the multiplier. For each embedding multiplier $\lambda\in\R$, the entropy-regularized
auxiliary problem is
\begin{equation}
\begin{aligned}
    \min_{\{\pi_t\}_{t=0}^{T-1}}\quad &
    \E_{0,y_0}\Bigg[
        \frac12y_T^2
        +
        \sum_{t=0}^{T-1}
        q_t^2y_t^2 \1_{\{y_t>0\}}
        \int_{\R_+^n}
        \pi_t(K)
        \log\pi_t(K)\,dK
    \Bigg].\\
    \text{s.t.}\quad & y_{t+1}= s_ty_t-s_t y_tP_t^\top K_t \1_{\{y_t>0\}},\\
    & K_t \sim \pi_t(K),\quad \pi(K)\ge 0, \quad \int_{\R_+^n}  \pi_t(K) dK = 1.
    \label{eq:regularized-auxiliary-problem}
\end{aligned}
\end{equation}
The auxiliary problem \eqref{eq:regularized-auxiliary-problem} can now
be solved by dynamic programming.

\begin{proposition}
\label{prop:soft-auxiliary-solution}
Suppose that $\tau$ lies in the convexity-preserving region of
Theorem~\ref{thm:whole-horizon-convexity-new}. Then
\begin{equation}
    J_t^\tau(y_t)
    =
    \frac12 C_t^\tau q_t^2y_t^2\1_{\{y_t>0\}}
    +
    \frac12 q_t^2y_t^2\1_{\{y_t\le0\}},
    \qquad
    t=0,1,\ldots,T.
    \label{eq:auxiliary-value-function}
\end{equation}
For $y_t>0$, the unique optimal density in the entropy-regularized
auxiliary problem is
\begin{equation}
    \pi_t^\tau(K)
    =
    \frac{
        \exp\!\left\{
            -h_t(K;C_{t+1}^\tau)/\tau
        \right\}
        \1_{\{K\in\R_+^n\}}
    }{
        \displaystyle
        \int_{\R_+^n}
        \exp\!\left\{
            -h_t(\nu;C_{t+1}^\tau)/\tau
        \right\}d\nu
    }.
    \label{eq:optimal-gibbs-policy}
\end{equation}
which is independent of the magnitude of
$y_t$. Its unique mode is
\begin{equation}
    K_t^{\tau,\sharp}
    =
    \argmin_{K\in\R_+^n}
    h_t(K;C_{t+1}^\tau).
    \label{eq:soft-mode}
\end{equation}
\end{proposition}

\begin{proof}
Under the structure-compatible temperature
\eqref{eq:state-dependent-temperature-new}, the Bellman equation of the
entropy-regularized auxiliary problem is
\begin{equation}
J_t^\tau(y_t)
= 
\min_{\pi_t}
\bigg\{
    \int_{\R_+^n}
    \E_{t,y_t}[J_{t+1}^\tau(y_{t+1}) ]\pi_t(K)\,dK
    +
    \tau
    \int_{\R_+^n}
    \pi_t(K)\log\pi_t(K)\,dK
\bigg\},
\label{eq:auxiliary-bellman-step}
\end{equation}
with terminal condition
\[
    J_T^\tau(y_T)=\frac12y_T^2.
\]
The minimization in the positive-state branch is over probability
densities $\pi_t$ supported on $\R_+^n$.

We prove \eqref{eq:auxiliary-value-function} by backward induction.
At $t=T$,
\[
    J_T^\tau(y_T)=\frac12y_T^2,
\]
which agrees with \eqref{eq:auxiliary-value-function} because
$q_T=1$ and $C_T^\tau=1$.

Suppose that \eqref{eq:auxiliary-value-function} holds at date $t+1$.
For $y_t\le0$, no randomized control is introduced
on this branch and the optimal risky position remains zero following the argument of \cite{cui2014optimal}. Hence
\[
    y_{t+1}=s_ty_t\le0.
\]
Using the induction hypothesis and $q_t=s_tq_{t+1}$,
\[
\begin{aligned}
    J_t^\tau(y_t)
    &=
    J_{t+1}^\tau(s_ty_t) =
    \frac12q_{t+1}^2s_t^2y_t^2
    =
    \frac12q_t^2y_t^2.
\end{aligned}
\]

Now consider $y_t>0$. For a given $K\in\R_+^n$,
\[
    u_t=s_ty_tK,
    \qquad
    y_{t+1}
    =
    s_ty_t(1-P_t^\top K).
\]
By the induction hypothesis,
\begin{align*}
    \E_{t,y_t}\!\left[
        J_{t+1}^\tau(y_{t+1})
    \right]
    &=
    \frac12q_t^2y_t^2
    \E\!\left[
        C_{t+1}^\tau
        (1-P_t^\top K)^2
        \1_{\{P_t^\top K<1\}}
        +
        (1-P_t^\top K)^2
        \1_{\{P_t^\top K\ge1\}}
    \right]=
    q_t^2y_t^2
    h_t(K;C_{t+1}^\tau).
\end{align*}
Since $\tau$ lies in the convexity-preserving region of
Theorem~\ref{thm:whole-horizon-convexity-new},
\[
    \int_{\R_+^n}
    \exp\!\left\{
        -\frac{h_t(K;C_{t+1}^\tau)}{\tau}
    \right\}dK
    <\infty.
\]
For any admissible density $\pi_t$ on $\R_+^n$, the entropy variational
identity gives
\begin{align*}
    &\int_{\R_+^n}
    h_t(K;C_{t+1}^\tau)\pi_t(K)\,dK
    +
    \tau
    \int_{\R_+^n}
    \pi_t(K)\log\pi_t(K)\,dK\\
    =&
    -\tau
    \log
    \int_{\R_+^n}
    \exp\!\left\{
        -\frac{h_t(K;C_{t+1}^\tau)}{\tau}
    \right\}dK
+
    \tau
    \int_{\R_+^n}
    \pi_t(K)
    \log
    \frac{\pi_t(K)}{\pi_t^\tau(K)}
    \,dK.
\end{align*}
The last term is the Kullback--Leibler divergence from $\pi_t$ to
$\pi_t^\tau$. It is nonnegative and equals zero if and only if
$\pi_t=\pi_t^\tau$ almost everywhere. Hence
\eqref{eq:optimal-gibbs-policy} is the unique minimizer of the
positive-state Bellman problem.

Using the definition
\[
    C_t^\tau
    =
    -2\tau
    \log
    \int_{\R_+^n}
    \exp\!\left\{
        -\frac{h_t(K;C_{t+1}^\tau)}{\tau}
    \right\}dK,
\]
the minimum value in
\eqref{eq:auxiliary-bellman-step} is
\[
    J_t^\tau(y_t)
    =
    \frac12C_t^\tau q_t^2y_t^2,
    \qquad
    y_t>0.
\]
Combining this result with the nonpositive-state branch gives
\eqref{eq:auxiliary-value-function} and completes the backward induction.

Finally, the denominator in
\eqref{eq:optimal-gibbs-policy} is independent of $K$, and the exponential
is strictly decreasing in $h_t(K;C_{t+1}^\tau)$. Hence the mode of
$\pi_t^\tau(K)$ coincides with
\[
    \argmin_{K\in\R_+^n}
    h_t(K;C_{t+1}^\tau).
\]
Theorem~\ref{thm:whole-horizon-convexity-new} gives strict convexity of
$h_t(\cdot;C_{t+1}^\tau)$, so the minimizer is unique. This proves
\eqref{eq:soft-mode}.
\end{proof}

For each fixed $\lambda$, Proposition~\ref{prop:soft-auxiliary-solution}
gives the optimal value of the entropy-regularized auxiliary problem with
initial transformed state
\[
    y_0=(d-\lambda)q_0^{-1}-x_0.
\]
Restoring the outer optimization over the embedding multiplier, we define
the regularized embedded problem by
\begin{equation}
    \lambda_\tau^*
    =
    \argmax_{\lambda\in\R}
    \left\{
        2J_0^\tau\!\left(
            (d-\lambda)q_0^{-1}-x_0
        \right)
        -
        \lambda^2
    \right\}.
    \label{eq:regularized-embedded-problem}
\end{equation}
Since
\[
    q_0y_0=d-\lambda-q_0x_0,
\]
\eqref{eq:auxiliary-value-function} implies that the objective in
\eqref{eq:regularized-embedded-problem} is
\[
    \begin{cases}
        C_0^\tau(d-\lambda-q_0x_0)^2-\lambda^2,
        &
        \lambda<d-q_0x_0,
        \\[1mm]
        (d-\lambda-q_0x_0)^2-\lambda^2,
        &
        \lambda\ge d-q_0x_0.
    \end{cases}
\]

\begin{theorem}
\label{thm:optimal-multiplier-feedback}
Suppose that the conditions of Theorem~\ref{thm:whole-horizon-convexity-new} hold, with $0<C_0<1$ and $d>q_0x_0$. Then there exists $\tau_0>0$ such that, for every $0<\tau<\tau_0$, $0<C_0^\tau<1$ and the regularized embedded problem \eqref{eq:regularized-embedded-problem} has the unique solution 
\begin{equation}
    \lambda_\tau^*
    =
    -\frac{C_0^\tau}{1-C_0^\tau}
    (d-q_0x_0).
    \label{eq:optimal-soft-multiplier}
\end{equation}
The corresponding optimal randomized portfolio policy is
\begin{equation}
    u_t^{\tau,*}
    =
    s_t
    \left(
        (d-\lambda_\tau^*)q_t^{-1}-x_t
    \right)_+
    K_t,
    \qquad
    K_t\sim\pi_t^\tau,
    \label{eq:optimal-randomized-policy}
\end{equation}
where $\pi_t^\tau$ is given by
\eqref{eq:optimal-gibbs-policy}. The associated deterministic mode
feedback is
\begin{equation}
    u_t^{\tau,\sharp}(x_t)
    =
    s_t
    \left(
        (d-\lambda_\tau^*)q_t^{-1}-x_t
    \right)_+
    K_t^{\tau,\sharp},
    \label{eq:optimal-mode-feedback}
\end{equation}
where $K_t^{\tau,\sharp}$ is defined in
\eqref{eq:soft-mode}.
\end{theorem}

\begin{proof}
By Theorem~\ref{thm:whole-horizon-convexity-new},
\[
    C_0^\tau\longrightarrow C_0
    \qquad
    \text{as }\tau\downarrow0.
\]
Since $0<C_0<1$, we have
\[
    0<C_0^\tau<1
\]
for every $0<\tau<\tau_0$.

On the branch $\lambda<d-q_0x_0$, the objective in
\eqref{eq:regularized-embedded-problem} is
\[
    C_0^\tau(d-\lambda-q_0x_0)^2-\lambda^2.
\]
Since $0<C_0^\tau<1$, this expression is strictly concave in $\lambda$.
Its unique maximizer is
\[
    \lambda_\tau^*
    =
    -\frac{C_0^\tau}{1-C_0^\tau}
    (d-q_0x_0),
\]
which is negative and therefore satisfies
$\lambda_\tau^*<d-q_0x_0$.

On the branch $\lambda\ge d-q_0x_0$, the objective becomes
\[
    (d-q_0x_0)^2
    -
    2(d-q_0x_0)\lambda,
\]
which is strictly decreasing because $d>q_0x_0$. Its maximum on this
branch is therefore attained at $\lambda=d-q_0x_0$ and equals
\[
    -(d-q_0x_0)^2.
\]
By contrast, the objective at $\lambda_\tau^*$ is
\[
    \frac{C_0^\tau}{1-C_0^\tau}
    (d-q_0x_0)^2>0.
\]
Hence $\lambda_\tau^*$ is the unique global solution of
\eqref{eq:regularized-embedded-problem}.

For this optimal multiplier, the transformed state is
\[
    y_t
    =
    (d-\lambda_\tau^*)q_t^{-1}-x_t.
\]
Proposition~\ref{prop:soft-auxiliary-solution} then implies that, whenever
$y_t>0$, the optimal coefficient satisfies
$K_t\sim\pi_t^\tau$, whereas the optimal risky position is zero for
$y_t\le0$. This gives \eqref{eq:optimal-randomized-policy}. Replacing
$K_t$ by its unique mode $K_t^{\tau,\sharp}$ gives
\eqref{eq:optimal-mode-feedback}.
\end{proof}

The multiplier $\lambda_\tau^*$ is chosen once at the initial date through
the outer problem \eqref{eq:regularized-embedded-problem} and is then held
fixed throughout the auxiliary dynamic program. At positive temperature,
it is the optimizer of the regularized embedding and need not enforce the
terminal mean constraint exactly. The original no-shorting solution is
recovered as $\tau\downarrow0$.

\begin{proposition}
\label{prop:zero-temperature-recovery}
Under the assumptions of
Theorem~\ref{thm:optimal-multiplier-feedback}, for every
$t=0,1,\ldots,T-1$,
\begin{equation*}
    K_t^{\tau,\sharp}
    \longrightarrow
    K_t^*,
    \qquad
    \tau\downarrow0,
\end{equation*}
and
\begin{equation*}
    \lambda_\tau^*
    \longrightarrow
    \lambda^*
    :=
    -\frac{C_0}{1-C_0}
    (d-q_0x_0).
\end{equation*}
Consequently, for every fixed wealth state $x_t$,
\begin{equation*}
    u_t^{\tau,\sharp}(x_t)
    \longrightarrow
    u_t^*(x_t),
\end{equation*}
where
\[
    u_t^*(x_t)
    =
    s_t
    \left(
        (d-\lambda^*)q_t^{-1}-x_t
    \right)_+
    K_t^* 
\]
is the optimal feedback policy of the original no-shorting problem and $(\cdot)_+ = \max{\{\cdot, 0\}}$.
\end{proposition}

\begin{proof}
By Theorem~\ref{thm:whole-horizon-convexity-new},
\[
    C_{t+1}^\tau
    \longrightarrow
    C_{t+1}.
\]
Hence
\[
    h_t(K;C_{t+1}^\tau)
    \longrightarrow
    h_t(K;C_{t+1})
\]
uniformly on every compact subset of $\R_+^n$. The common coercive bound
established in Theorem~\ref{thm:whole-horizon-convexity-new} keeps the
minimizers in a common compact set for all sufficiently small $\tau$.
Since
\[
    K_t^{\tau,\sharp}
    =
    \argmin_{K\in\R_+^n}
    h_t(K;C_{t+1}^\tau),
\]
and $h_t(\cdot;C_{t+1})$ is strictly convex with unique minimizer $K_t^*$,
standard argmin convergence gives
\[
    K_t^{\tau,\sharp}
    \longrightarrow
    K_t^*.
\]

Next, \eqref{eq:optimal-soft-multiplier} and
$C_0^\tau\to C_0\in(0,1)$ imply
\[
    \lambda_\tau^*
    =
    -\frac{C_0^\tau}{1-C_0^\tau}
    (d-q_0x_0)
    \longrightarrow
    -\frac{C_0}{1-C_0}
    (d-q_0x_0)
    =
    \lambda^*.
\]
Therefore,
\[
    (d-\lambda_\tau^*)q_t^{-1}
    \longrightarrow
    (d-\lambda^*)q_t^{-1}.
\]
Combining this convergence with
$K_t^{\tau,\sharp}\to K_t^*$ and continuity of the positive-part mapping
gives, for every fixed $x_t$,
\[
\begin{aligned}
    u_t^{\tau,\sharp}(x_t)
    &=
    s_t
    \left(
        (d-\lambda_\tau^*)q_t^{-1}-x_t
    \right)_+
    K_t^{\tau,\sharp}
\longrightarrow
    s_t
    \left(
        (d-\lambda^*)q_t^{-1}-x_t
    \right)_+
    K_t^*
    =
    u_t^*(x_t).
\end{aligned}
\]
\end{proof}

\section{Sample-Based Numerical Implementation}
\label{sec:sample-implementation} 

The recursion in Section~\ref{sec:structure-preserving-soft} depends on expectations with respect to the
distribution of risky excess returns and on partition integrals over
$\R_+^n$. When the return distribution is unknown, we replace the former
by sample averages and evaluate the latter numerically. 

At date $t$, let
\[
    \{P_t^{(j)}\}_{j=1}^{M_t}
\]
denote the available sample of risky excess returns. For a continuation
coefficient $c>0$, define the empirical counterpart of
$h_t(K;c)$ on $\R_+^n $ by
\begin{equation}
\begin{aligned}
    \widehat h_t(K;c)
    :=
    \frac{1}{2M_t}
    \sum_{j=1}^{M_t}
    \Big[
        &c\bigl(1-(P_t^{(j)})^\top K\bigr)^2
        \1_{\{(P_t^{(j)})^\top K<1\}}+
        \bigl(1-(P_t^{(j)})^\top K\bigr)^2
        \1_{\{(P_t^{(j)})^\top K\ge1\}}
    \Big].
\end{aligned}
\label{eq:empirical-stage-energy}
\end{equation}
Starting from
\(
    \widehat C_T^\tau=1,
\)
the sample analogue of the structure-preserving recursion is
\begin{equation}
    \widehat C_t^\tau
    =
    -2\tau
    \log
    \int_{\R_+^n}
    \exp\!\left\{
        -\frac{
            \widehat h_t(K;\widehat C_{t+1}^\tau)
        }{\tau}
    \right\}dK,
    \qquad
    t=T-1,\ldots,0,
\label{eq:empirical-soft-recursion}
\end{equation}
whenever the integral is finite. The corresponding empirical mode is
\begin{equation}
    \widehat K_t^{\tau,\sharp}
    =
    \argmin_{K\in\R_+^n}
    \widehat h_t(K;\widehat C_{t+1}^\tau).
\label{eq:empirical-mode}
\end{equation}
Conditional on $\widehat C_{t+1}^\tau>0$, the empirical stage function is
convex; when the empirical second-moment matrix is positive definite, it
is strictly convex and coercive. Hence
\eqref{eq:empirical-mode} can be solved by standard constrained convex
optimization.

The remaining numerical task is the integral in
\eqref{eq:empirical-soft-recursion}. For small $\tau$, its mass is
concentrated around $\widehat K_t^{\tau,\sharp}$. We therefore use a
lower-truncated Gaussian importance distribution centered at this mode.
Define the local curvature matrix
\[
    \widehat G_t
    =
    \frac{1}{M_t}
    \sum_{j=1}^{M_t}
    \widehat\omega_{t,j}
    P_t^{(j)}(P_t^{(j)})^\top,
\]
where
\[
    \widehat\omega_{t,j}
    =
    \widehat C_{t+1}^\tau
    \1_{\{(P_t^{(j)})^\top\widehat K_t^{\tau,\sharp}<1\}}
    +
    \1_{\{(P_t^{(j)})^\top\widehat K_t^{\tau,\sharp}\ge1\}}.
\]
For numerical stability, let
\[
    \widehat G_{t,\epsilon}
    =
    \widehat G_t+\epsilon I_n,
    \qquad
    \Gamma_t
    =
    \tau\widehat G_{t,\epsilon}^{-1},
\]
where $\epsilon>0$ is a small ridge parameter. Let $\psi_t$ denote the
$N(\widehat K_t^{\tau,\sharp},\Gamma_t)$ density truncated to
$\R_+^n$. Using $N_q$ randomized Sobol points transported through $\psi_t$, denote
the resulting draws by
\[
    K_t^{(\ell)},
    \qquad
    \ell=1,\ldots,N_q.
\]
The log-partition integral is estimated directly on the log scale as
\begin{equation}
\begin{aligned}
    &\log
    \int_{\R_+^n}
    \exp\!\left\{
        -\frac{
            \widehat h_t(K;\widehat C_{t+1}^\tau)
        }{\tau}
    \right\}dK
 \approx
    -\log N_q
    +
    \operatorname{LSE}_{1\le\ell\le N_q}
    \left[
        -\frac{
            \widehat h_t(
                K_t^{(\ell)};
                \widehat C_{t+1}^\tau
            )
        }{\tau}
        -
        \log\psi_t(K_t^{(\ell)})
    \right],
\end{aligned}
\label{eq:empirical-logpartition}
\end{equation}
where
\[
    \operatorname{LSE}(a_1,\ldots,a_m)
    :=
    \log\!\left(\sum_{\ell=1}^m e^{a_\ell}\right)
\]
is evaluated using the standard log-sum-exp stabilization.

The positivity condition in Theorem~\ref{thm:whole-horizon-convexity-new}
is imposed directly on the sample recursion. If
$\widehat C_t^\tau\le0$ at any date, the backward recursion is stopped
and the temperature is reduced.  After a successful backward recursion, the sample plug-in counterpart of
\eqref{eq:optimal-soft-multiplier} is
\begin{equation}
    \widehat\lambda_\tau^*
    =
    -\frac{\widehat C_0^\tau}
    {1-\widehat C_0^\tau}
    (d-q_0x_0).
\label{eq:empirical-multiplier}
\end{equation}
The deterministic mode policy used in the numerical and empirical
analysis is therefore
\begin{equation}
    \widehat u_t^{\tau,\sharp}(x_t)
    =
    s_t
    \left(
        (d-\widehat\lambda_\tau^*)q_t^{-1}-x_t
    \right)_+
    \widehat K_t^{\tau,\sharp}.
\label{eq:empirical-mode-policy}
\end{equation}
If randomized execution is desired, the corresponding empirical Gibbs
density is obtained by replacing $h_t$ and $C_{t+1}^\tau$ in
\eqref{eq:optimal-gibbs-policy} with
$\widehat h_t$ and $\widehat C_{t+1}^\tau$, respectively.

The complete procedure is summarized in Algorithm~\ref{alg:sample-soft}.
At each date, it requires only one constrained coefficient minimization
and one log-partition calculation; no discretization of the wealth state
is required.

\begin{algorithm}[ht]
\caption{Sample-based structure-preserving recursion}
\label{alg:sample-soft}
\begin{algorithmic}[1]
\Require Samples $\{P_t^{(j)}\}_{j=1}^{M_t}$,
         $\{s_t\}_{t=0}^{T-1}$, $x_0,d,\tau,N_q,\epsilon$.
\State Set $\widehat C_T^\tau\gets1$.
\For{$t=T-1,\ldots,0$}
    \State Construct
    $\widehat h_t(K;\widehat C_{t+1}^\tau)$ from
    \eqref{eq:empirical-stage-energy}.
    \State Compute
    $\widehat K_t^{\tau,\sharp}
    \gets
    \argmin_{K\in\R_+^n}
    \widehat h_t(K;\widehat C_{t+1}^\tau)$ by standard constrained convex
optimization algorithm.
    \State Form $\widehat G_{t,\epsilon}$ and the truncated Gaussian
    proposal $\psi_t$ centered at $\widehat K_t^{\tau,\sharp}$.
    \State Estimate the log-partition integral by randomized Sobol
quasi-Monte Carlo (QMC) importance sampling using
\eqref{eq:empirical-logpartition}.
    \State Update $\widehat C_t^\tau$ using
    \eqref{eq:empirical-soft-recursion}.
    \If{$\widehat C_t^\tau\le0$}
        \State Stop; reduce $\tau$ and restart from
        $\widehat C_T^\tau=1$.
    \EndIf
\EndFor
\If{$0<\widehat C_0^\tau<1$ and $d>q_0x_0$}
    \State Compute $\widehat\lambda_\tau^*$ from
    \eqref{eq:empirical-multiplier} and execute
    \eqref{eq:empirical-mode-policy}.
\Else
    \State The closed-form embedding solution is not applicable.
\EndIf
\end{algorithmic}
\end{algorithm}

\section{Numerical Analysis}
\label{sec:numerical-analysis}

The numerical analysis examines the two main implications of the theory.
We first verify that, under no-shorting, a fixed absolute entropy
temperature generally destroys the quadratic Bellman closure, whereas the
state-dependent temperature in
\eqref{eq:state-dependent-temperature-new} preserves it. We then examine
the whole-horizon positivity condition in
Theorem~\ref{thm:whole-horizon-convexity-new} and evaluate portfolio
performance within the resulting admissible temperature region.

We use a two-regime Gaussian-mixture design with $n=4$ risky assets and
$T=6$ investment periods. The risky excess-return vector is generated from
a two-component lognormal mixture:
\[
    P_t+s_t\mathbf 1
    \sim
    0.82\,\mathrm{LN}(\mu_N,\Sigma_N)
    +
    0.18\,\mathrm{LN}(\mu_S,\Sigma_S),
    \qquad
    s_t=1.005,
\]
where $\mathrm{LN}(\mu,\Sigma)$ denotes the multivariate lognormal
distribution whose logarithm is Gaussian with mean $\mu$ and covariance
matrix $\Sigma$.

The regime-specific mean vectors are
\begin{align*}
    &\mu_N
    =
    (0.036332,\;0.029559,\;0.023717,\;0.040182)^\top,\\
      &  \mu_S
    =
    (-0.058668,\;-0.065441,\;-0.071283,\;-0.054818)^\top.
\end{align*}
The covariance matrices are
\[
    \Sigma_N=D_NR D_N,
    \qquad
    \Sigma_S=D_SR D_S,
\]
with
\[
    D_N
    =
    \operatorname{diag}
    (0.050,\;0.065,\;0.080,\;0.055),
    \quad
    D_S=2.25D_N,\quad     R
    =
    \begin{pmatrix}
        1    & 0.55 & 0.55 & 0.55\\
        0.55 & 1    & 0.55 & 0.55\\
        0.55 & 0.55 & 1    & 0.55\\
        0.55 & 0.55 & 0.55 & 1
    \end{pmatrix}.
\]
Initial wealth is $x_0=1$, the target terminal mean is $d=1.18$, and the
training and testing samples contain $5{,}000$ and $10{,}000$
independent paths, respectively.

For the structural comparison, we focus on the final Bellman stage, where
the continuation value is known exactly. We consider the unconstrained
fixed-temperature formulation, the no-shorting fixed-temperature
formulation, and the proposed state-dependent formulation. Their
date-$T-1$ value functions are obtained, respectively, from the
corresponding Bellman equations. The two fixed-temperature formulations
are evaluated at $\bar\tau=10^{-3}$, while the proposed formulation uses
the base temperature $\tau=10^{-3}$. 

\begin{figure}[htbp]
    \centering
    \includegraphics[width=0.75\textwidth]
    {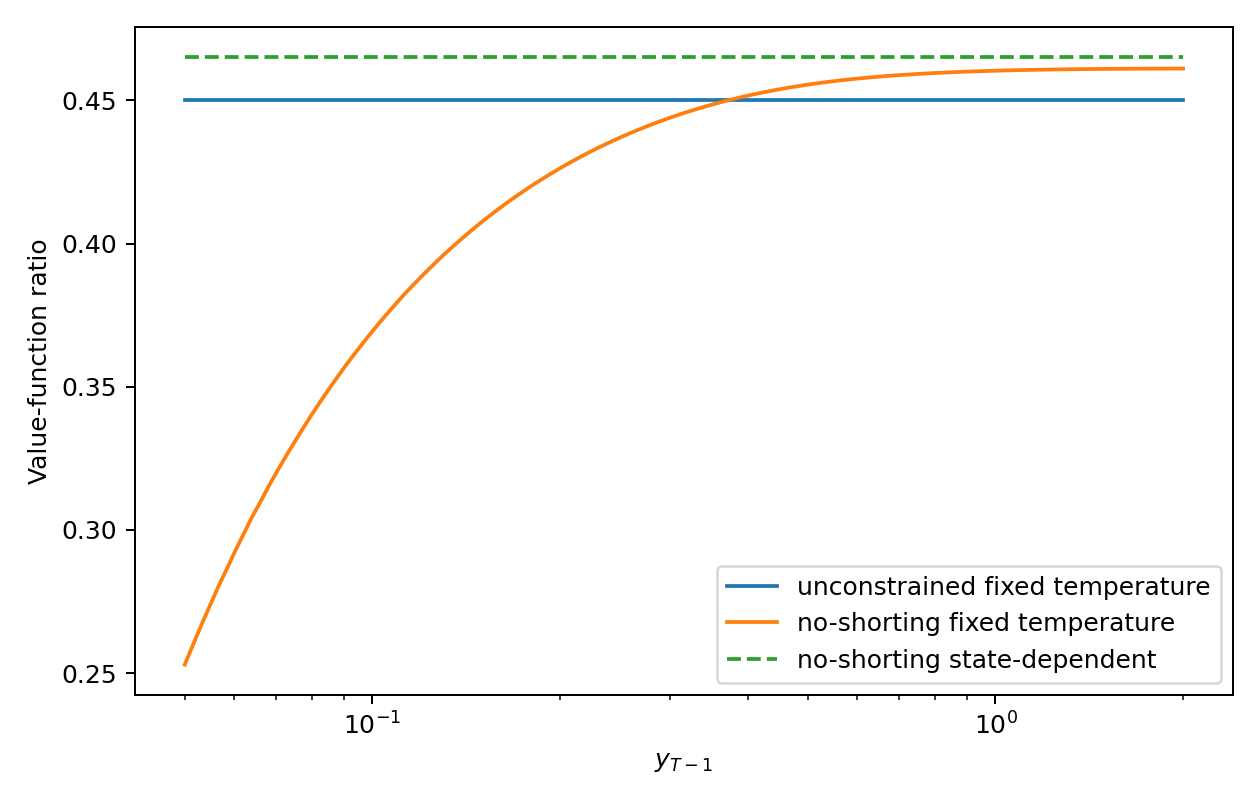}
 \caption{Bellman structure under fixed and state-dependent entropy
regularization. The figure plots, for each formulation, the date-$T-1$
value function divided by $q_{T-1}^2y_{T-1}^2$ over $y_{T-1}>0$.
The two fixed-temperature formulations use $\bar\tau=10^{-3}$, while
the proposed state-dependent formulation uses $\tau=10^{-3}$. A constant
ratio indicates a quadratic positive-state value function.}
    \label{fig:structure-preservation}
\end{figure}

Figure~\ref{fig:structure-preservation} plots, for each formulation, the
date-$T-1$ value function divided by the quadratic scale
$q_{T-1}^2y_{T-1}^2$, over $y_{T-1}>0$. If the positive-state value
function is quadratic in $y_{T-1}$, this ratio must be constant.  
The ratio is constant at approximately $0.4502$ for the unconstrained
fixed-temperature benchmark and $0.4650$ for the proposed
state-dependent formulation. In contrast, under no-shorting with a fixed
absolute temperature, it increases from approximately $0.2529$ at
$y_{T-1}=0.05$ to $0.4611$ at $y_{T-1}=2$. Hence the first and third
formulations retain a quadratic positive-state value function, whereas
the fixed-temperature no-shorting formulation does not. The deviation is
largest near the truncation boundary and diminishes as $y_{T-1}$
increases, when the boundary becomes less relevant.

We next examine whether the structure-preserving recursion remains
convex over the entire investment horizon. We consider the base-temperature grid
\begin{align*}
    \tau\in
    \{&
    10^{-4},
    2.5\times10^{-4},
    5\times10^{-4},
    7.5\times10^{-4},
    10^{-3},
    2.5\times10^{-3},\\
    &5\times10^{-3},
    7.5\times10^{-3},
    10^{-2},
    1.5\times10^{-2},
    2.25\times10^{-2}
    \}.
\end{align*}
For each value of $\tau$, we apply the sample-based backward recursion in
Section~\ref{sec:sample-implementation} and record the smallest empirical
continuation coefficient generated along the recursion. If a nonpositive
coefficient is encountered, the recursion is terminated at that stage,
and that coefficient is recorded as the minimum for the corresponding
temperature. Figure~\ref{fig:min-continuation} plots this minimum
continuation coefficient against $\tau$.
\begin{figure}[htbp]
    \centering
    \includegraphics[width=0.75\textwidth]
    {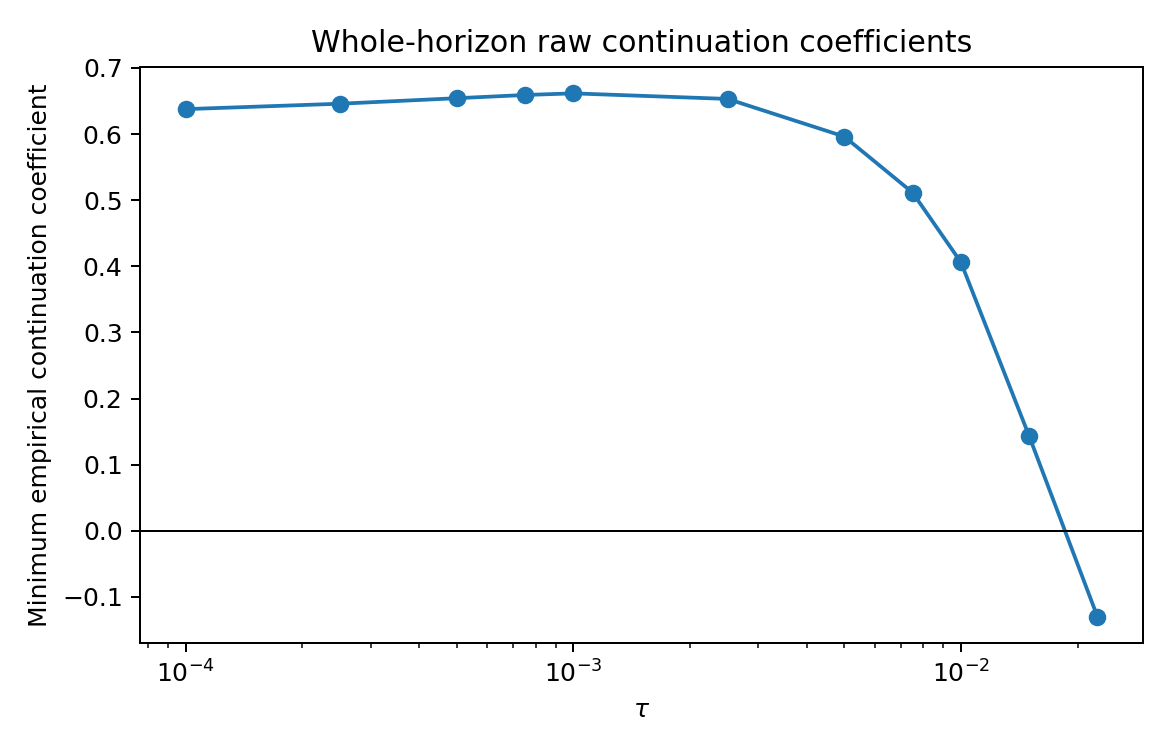}
\caption{Whole-horizon continuation-coefficient diagnostic.
For each base temperature $\tau$, the figure reports the smallest
empirical continuation coefficient generated along the backward
recursion. The recursion stops when a nonpositive coefficient first
appears, and that coefficient is recorded as the minimum. The horizontal
line marks the convexity boundary at zero.}
    \label{fig:min-continuation}
\end{figure}

The recursion remains in the positive-coefficient region through
$\tau=0.015$. At small temperatures, the minimum coefficient rises
slightly from $0.6376$ at $\tau=10^{-4}$ to $0.6615$ at
$\tau=10^{-3}$. It then declines to $0.5962$ at $\tau=0.005$,
$0.4061$ at $\tau=0.01$, and $0.1432$ at $\tau=0.015$. Thus the effect
of entropy on the continuation coefficients is not monotone even within
the admissible region, but sufficiently large regularization eventually
pushes the recursion toward the curvature boundary. At $\tau=0.0225$, the recursion produces
\[
    \widehat C_1^\tau=-0.1298,
\]
so $\widehat C_0^\tau$ is not formed under the convexity-preserving
recursion. This result gives a direct numerical counterpart to
Theorem~\ref{thm:whole-horizon-convexity-new}: the admissible set contains
a nontrivial range of strictly positive temperatures, but positivity is
not automatic for arbitrary $\tau$. Once one continuation coefficient
becomes nonpositive, the preceding stage can no longer be continued as a
convex coefficient problem. For all reported temperatures at which $\widehat C_0^\tau$ remains
positive, it is the smallest continuation coefficient. This indicates
that the initial coefficient is typically the binding one, reflecting
the cumulative effect of entropy corrections along the backward
recursion. The case $\tau=0.0225$, however, shows that the breakdown need
not occur at $t=0$: $\widehat C_1^\tau$ becomes negative before
$\widehat C_0^\tau$ can be computed. Hence positivity is a genuinely
recursive requirement, not merely a condition on the initial
continuation coefficient.

We next evaluate the economic implications of the deterministic mode
policy. Table~\ref{tab:simulation-performance-new} compares two
estimate-then-optimize benchmarks with the proposed deterministic mode
policy at three representative temperatures. ETO-GMM uses the correctly
specified two-component full-covariance Gaussian-mixture model, whereas
ETO-GD uses a deliberately misspecified single-regime diagonal Gaussian
approximation. The deterministic mode policy in (\ref{eq:empirical-mode-policy})
 uses the sample plug-in
multiplier $\widehat\lambda_\tau^*$ and the empirical mode coefficients
$\widehat K_t^{\tau,\sharp}$ obtained from Section~4.

\begin{table}[htbp]
    \centering
    \caption{Out-of-sample portfolio performance within the
    convexity-preserving temperature region}
    \label{tab:simulation-performance-new}
    \begin{tabular}{lrrrrr}
        \toprule
        Method & $\tau$
        & \shortstack{Terminal\\mean}
        & \shortstack{Terminal\\variance}
        & \shortstack{Target\\error}
        & \shortstack{Sharpe\\ratio} \\
        \midrule
        ETO-GMM & -- & 1.1758 & 0.0404 & -0.0042 & 0.7235 \\
        ETO-GD  & -- & 1.1766 & 0.0805 & -0.0034 & 0.5152 \\
        Deterministic mode & $10^{-4}$ & 1.1820 & 0.0441 & 0.0020 & 0.7220 \\
        Deterministic mode & $10^{-3}$ & 1.1926 & 0.0505 & 0.0126 & 0.7220 \\
        Deterministic mode & $10^{-2}$ & 1.1228 & 0.0164 & -0.0572 & 0.7220 \\
        \bottomrule
    \end{tabular}
\end{table}

ETO-GMM and ETO-GD deliver similar terminal means, but the misspecified
ETO-GD has substantially higher variance and a lower Sharpe-like ratio.
At $\tau=10^{-4}$, the deterministic mode policy closely tracks the
target and delivers risk-adjusted performance comparable to the correctly
specified ETO-GMM benchmark. As $\tau$ increases, the economic effects become more pronounced. At
$\tau=10^{-3}$ the policy overshoots the target, while at
$\tau=10^{-2}$ it becomes much more conservative and substantially
undershoots it, even though the temperature remains in the
convexity-preserving region. Thus, positivity of the continuation
coefficients guarantees structural validity of the recursion, but does
not by itself determine an economically desirable temperature.

Finally, we examine whether the continuation-coefficient results are
sensitive to numerical evaluation of the log-partition integrals.
Table~\ref{tab:qmc-sensitivity} reports the recursion at $\tau=10^{-3}$
using $2{,}048$, $4{,}096$, and $8{,}192$ randomized QMC points. The
baseline specification uses $N_q=4{,}096$. The estimates are essentially unchanged across QMC resolutions. Relative
to the $4{,}096$-point baseline, the maximum relative deviation of
$\widehat C_0^\tau$ is below $0.01\%$, while
$\widehat\lambda_\tau^*$ and the resulting terminal moments vary only in
the fourth or later decimal places. The continuation-coefficient pattern
in Figure~\ref{fig:min-continuation} is therefore not driven by the
chosen QMC resolution.

\begin{table}[htbp]
    \centering
    \caption{Sensitivity to the number of quasi-Monte Carlo (QMC) points}
    \label{tab:qmc-sensitivity}
    \begin{tabular}{rrrrrr}
        \toprule
        $N_q$ & $\widehat C_0^\tau$
        & $\min_t\widehat C_t^\tau$
        & $\widehat\lambda_\tau^*$
        & \shortstack{Terminal\\mean}
        & \shortstack{Terminal\\variance} \\
        \midrule
        2,048 & 0.661525 & 0.661525 & -0.292427
              & 1.192672 & 0.050530 \\
        4,096 & 0.661462 & 0.661462 & -0.292344
              & 1.192641 & 0.050511 \\
        8,192 & 0.661496 & 0.661496 & -0.292388
              & 1.192658 & 0.050521 \\
        \bottomrule
    \end{tabular}
\end{table}

Overall, the numerical evidence mirrors the two theoretical results of
the paper. No-shorting breaks the quadratic closure of the standard
fixed-temperature formulation, while the state-dependent temperature
restores the finite-dimensional Bellman structure. This restoration is
valid only over a finite positive-temperature region in which all
continuation coefficients remain positive. Within that region, the
sample-based mode policy produces economically meaningful portfolio
outcomes, but structural admissibility alone does not determine the
preferred temperature.

\section{Empirical Illustration}
\label{sec:empirical-illustration}

We examine whether the continuation-coefficient restriction identified
in Section~\ref{sec:structure-preserving-soft} remains relevant in
real-data implementation and how the deterministic mode policy performs
out of sample relative to standard estimate-then-optimize benchmarks.

The data are obtained from the Kenneth French Data Library and are
monthly. The risky assets are ten
Fama--French industry portfolios: NoDur, Durbl, Manuf, Enrgy, HiTec,
Telcm, Shops, Hlth, Utils, and Other, whose monthly return are computed by value-weighted method. We also use the monthly market
excess return, SMB, HML, and the one-month risk-free rate. The sample
runs from January 1975 to December 2025. All returns are converted from
percentages to decimals. Within each rolling window, the riskless gross
return $s_t$ used in the investment is fixed at the average
one-month risk-free gross return over the training sample, and the risky
excess-return vectors $P_t$ are constructed relative to this value.

Each portfolio problem starts from $x_0=1$, has target terminal mean
$d=1.1$, and uses a six-month horizon, $T=6$. At each evaluation date,
the preceding $180$ months are used for estimation, and the resulting
policy is applied to the following non-overlapping six-month test
episode; the window is then advanced by six months. The first test
episode begins in January 1990 and the last in July 2025, yielding $72$ non-overlapping out-of-sample episodes. For the deterministic mode policy, each
$180$-month training window is converted into $175=180-6+1$ overlapping
six-month return paths, which enter the sample-based recursion in
Section~\ref{sec:sample-implementation} directly. We compare this policy
with three estimate-then-optimize benchmarks. ETO-GMM selects, by BIC,
a Gaussian-mixture specification with one to three components and full
or diagonal covariance matrices; ETO-AR selects the autoregressive lag
order by BIC from $p\in\{0,1,2,3,4,6\}$; and ETO-FF3 models industry
returns using the Fama--French three factors. Each fitted benchmark model
generates six-month return paths that are passed through the same
deterministic no-shorting Bellman recursion. Thus all methods use the
same rolling windows, investment horizon, mean--variance objective, and
no-shorting constraint, while differing in how the return distribution
enters the portfolio problem.

Figure~\ref{fig:empirical-admissibility} reports the rolling counterpart
of the whole-horizon continuation-coefficient diagnostic over the base-temperature grid
\[
    \tau\in
    \left\{
    10^{-4},
    2.5\times10^{-4},
    5\times10^{-4},
    7.5\times10^{-4},
    10^{-3},
    2.5\times10^{-3},
    5\times10^{-3},
    7.5\times10^{-3},
    10^{-2}
    \right\}.
\]
For each $\tau$ and each of the $72$ rolling training windows, we run the
backward recursion and record the smallest empirical continuation
coefficient generated within that window. Panel~(a) reports the fraction
of windows in which all continuation coefficients remain positive.
Panel~(b) summarizes the distribution of the resulting $72$ window-level
minimum coefficients by reporting, for each $\tau$, their $5\%$ quantile
and their overall minimum.

\begin{figure}[htbp]
    \centering
    \includegraphics[width=0.96\textwidth]
    {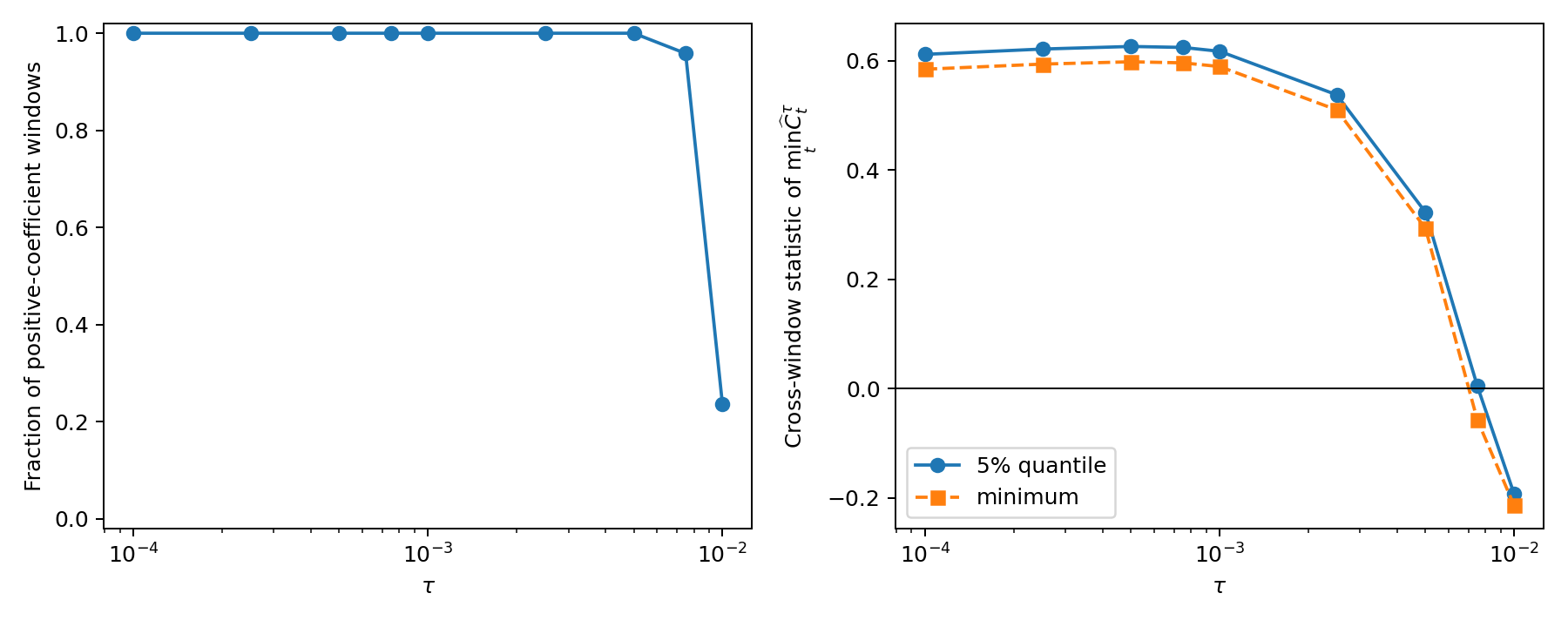}
\caption{Rolling continuation-coefficient diagnostic.
For each rolling window, the backward recursion stops when a nonpositive
continuation coefficient first appears. Panel~(a) reports the fraction of
the $72$ windows with all positive coefficients, and Panel~(b) reports
the $5\%$ quantile and minimum of the within-window minima.}
    \label{fig:empirical-admissibility}
\end{figure}

The empirical pattern closely mirrors the simulation evidence in
Section~\ref{sec:numerical-analysis}. All $72$ rolling windows remain in
the positive-coefficient region through $\tau=5\times10^{-3}$. At
$\tau=7.5\times10^{-3}$, $69$ windows, or $95.8\%$, remain positive, but
the minimum coefficient across all windows has already fallen to about
$-0.0574$, while the $5\%$ quantile is only $0.0049$. At
$\tau=10^{-2}$, only $17$ of the $72$ windows remain in the
positive-coefficient region and the overall minimum falls to about
$-0.2137$. Thus the convexity boundary identified theoretically is also
economically relevant in rolling real-data applications: sufficiently
large entropy regularization causes the backward recursion to fail in a
substantial fraction of samples.

We next compare portfolio outcomes at the temperatures for which all
$72$ rolling windows remain in the positive-coefficient region.
Table~\ref{tab:empirical-performance-d110} reports the out-of-sample
performance of the deterministic mode policy and the three ETO
benchmarks.

\begin{table}[htbp]
\centering
\small
\caption{Out-of-sample performance in rolling Fama--French industry portfolios ($d=1.10$)}
\label{tab:empirical-performance-d110}
\begin{tabular}{lrrrrrrr}
\toprule
Method & $\tau$ & \shortstack{Terminal\\mean} & \shortstack{Terminal\\variance} & \shortstack{Target\\error} & $q_{0.05}$ & \shortstack{Wealth\\CVaR$_{0.05}$} & \shortstack{Sharpe\\ratio} \\
\midrule
Deterministic mode & 0.00010 & 1.0744 & 0.0494 & -0.0256 & 0.7466 & 0.3568 & 0.2458 \\
Deterministic mode & 0.00025 & 1.0798 & 0.0555 & -0.0202 & 0.7292 & 0.3289 & 0.2544 \\
Deterministic mode  & 0.00050 & 1.0856 & 0.0618 & -0.0144 & 0.7122 & 0.3074 & 0.2646 \\
Deterministic mode  & 0.00075 & 1.0895 & 0.0648 & -0.0105 & 0.7054 & 0.3055 & 0.2737 \\
Deterministic mode  & 0.00100 & 1.0909 & 0.0648 & -0.0091 & 0.7055 & 0.3177 & 0.2791 \\
Deterministic mode  & 0.00250 & 1.0766 & 0.0407 & -0.0234 & 0.7704 & 0.4608 & 0.2814 \\
Deterministic mode  & 0.00500 & 1.0498 & 0.0135 & -0.0502 & 0.8740 & 0.6737 & 0.2585 \\
ETO-GMM & -- & 1.0692 & 0.0454 & -0.0308 & 0.8257 & 0.4281 & 0.2320 \\
ETO-AR & -- & 1.0701 & 0.0415 & -0.0299 & 0.8092 & 0.4629 & 0.2467 \\
ETO-FF3 & -- & 1.0691 & 0.0393 & -0.0309 & 0.8261 & 0.4782 & 0.2488 \\
\bottomrule
\end{tabular}
\end{table}

The ETO benchmarks deliver very similar terminal means and
risk profiles. In comparison, the deterministic mode policy exhibits a
clear temperature-dependent trade-off. As $\tau$ increases from
$10^{-4}$ to $10^{-3}$, the terminal mean moves closer to the target
$d=1.10$ and the Sharpe ratio rises from $0.2458$ to $0.2791$, although
terminal variance also increases. At $\tau=10^{-3}$, the target error is
reduced to $-0.0091$, substantially smaller in magnitude than those of
the three ETO benchmarks, which are around $-0.03$. For larger temperatures, the policy becomes increasingly conservative.
At $\tau=2.5\times10^{-3}$, variance falls to $0.0407$ and downside
wealth improves, while the target error widens to $-0.0234$. At
$\tau=5\times10^{-3}$, variance and downside risk improve markedly, but
the terminal mean falls to $1.0498$, producing a target error of
$-0.0502$. Thus, even within the positive-coefficient region, the
temperature materially changes the balance between target tracking and
risk reduction; recursive convexity alone does not determine the
economically preferred temperature.

The empirical evidence therefore reinforces the main theoretical
distinction. The positivity restriction is operationally relevant in
real data, since larger temperatures eventually move many rolling
problems outside the convex regime. Within the positive-coefficient
region, the deterministic mode policy remains implementable and produces
economically reasonable outcomes, but recursive convexity alone does not
select the preferred temperature.

\section{Conclusion}
\label{sec:conclusion}

This paper studies entropy regularization in the discrete-time
mean--variance problem under no-shorting. A standard fixed absolute
temperature can preserve convexity but generally destroys the homogeneous
Bellman structure that makes the constrained problem tractable. We show
that matching the entropy scale to the quadratic wealth-gap structure
restores the piecewise-quadratic value function and the finite-dimensional
backward recursion. The resulting continuation coefficients, however,
must remain positive to preserve recursive convexity. We prove that a
common sufficiently small base temperature guarantees 
convexity of the entire finite-horizon recursion, and we establish
zero-temperature recovery of the original no-shorting solution. The
sample-based implementation, numerical analysis, and rolling
Fama--French industry-portfolio application confirm both implications:
the state-dependent temperature preserves Bellman structure, while the
positivity restriction defines a genuine finite temperature region.
Within that region, the deterministic mode policy remains implementable,
although convexity alone does not determine the economically preferred
temperature.

Several extensions are natural. One is to develop data-driven temperature
selection subject to the recursive positivity constraint. Another is to
improve the scalability of the log-partition calculation for portfolios
with many risky assets. More broadly, the structure--convexity trade-off
identified here may arise in other constrained stochastic control
problems whose tractability relies on homogeneity or another
finite-dimensional Bellman structure.

\bibliographystyle{plainnat}
\bibliography{references}

\end{document}